\documentclass[format=acmsmall, review=false]{acmart}
\usepackage{acm-ec-25}
\usepackage{booktabs} % For formal tables
\usepackage[ruled]{algorithm2e} % For algorithms

\SetAlFnt{\small}
\SetAlCapFnt{\small}
\SetAlCapNameFnt{\small}
\SetAlCapHSkip{0pt}
\IncMargin{-\parindent}

\usepackage{thm-restate}

\usepackage[utf8]{inputenc}
\usepackage{bm}
\usepackage{bbm}

\newcommand{\bbN}{\mathbb{N}}

\newcommand{\bbR}{\mathbb{R}}

\newcommand{\bbZ}{\mathbb{Z}}

\newcommand{\mcD}{\mathcal{D}}

\newcommand{\mcF}{\mathcal{F}}

\newcommand{\mcK}{\mathcal{K}}
\newcommand{\mcL}{\mathcal{L}}

\newcommand{\mcP}{\mathcal{P}}

\newcommand{\mcS}{\mathcal{S}}
\newcommand{\mcT}{\mathcal{T}}

\newcommand{\inN}{\in\bbN}

\newcommand{\set}[1]{\left\{#1\right\}}

\newcommand{\cross}{\times}

\newcommand{\mb}{\mathbb{P}}
\newcommand{\me}{\mathbb{E}}
\newcommand{\mbb}[1]{\mb\left(#1\right)}
\newcommand{\meb}[1]{\me\left(#1\right)}

\newcommand{\indc}{\mathbbm{1}}

\newcommand{\given}{\mid}

\DeclareMathOperator{\distributed}{\thicksim}

\DeclareMathOperator{\Exp}{Exp}

\DeclareMathOperator{\Pois}{Pois}

\newtheorem{claim}{Claim}
\newcommand{\eps}{\varepsilon}
\newcommand{\dta}{\delta}

\newcommand{\closure}[1]{\overline{#1}}

\DeclareMathOperator{\balloperator}{B}
\newcommand{\ball}[2]{\balloperator_{#1}\kern-0.5ex\br{#2}}
\newcommand{\balld}[2]{\balloperator_{#1}\kern-1.0ex'\kern-0.5ex\br{#2}}
\newcommand{\cball}[2]{\closure{\balloperator}_{#1}\kern-0.5ex\br{#2}}
\newcommand{\cballd}[2]{\closure{\balloperator}_{#1}\kern-1.0ex'\kern-0.5ex\br{#2}}

\newcommand{\abs}[1]{\left\lvert#1\right\rvert}
\newcommand{\norm}[1]{\left\lVert#1\right\rVert}
\newcommand{\ceil}[1]{\left\lceil#1\right\rceil}
\newcommand{\floor}[1]{\left\lfloor#1\right\rfloor}

\newcommand{\br}[1]{\left(#1\right)}
\newcommand{\sbr}[1]{\left[#1\right]}

\DeclareMathOperator{\sech}{sech}

\newcommand{\piecewise}[1]{\left\{\begin{array}{ll}#1\end{array}\right.}

\usepackage[font=footnotesize]{caption}
\usepackage{xcolor}
\usepackage{etoolbox}
\usepackage[normalem]{ulem}

\usepackage{comment}
\setcitestyle{authoryear}

\title[Longer Lists Yield Better Matchings]{Longer Lists Yield Better Matchings}

\author{Yuri Faenza and Aapeli Vuorinen, IEOR, Columbia University}

\begin{abstract}
Many centralized mechanisms for two-sided matching markets that enjoy strong theoretical properties assume that the planner solicits full information on the preferences of each participating agent. In particular, they expect that participants compile and communicate their complete preference lists over agents from the other side of the market. However, real-world markets are often very large and agents cannot always be expected to even produce a ranking of all options on the other side. It is therefore important to understand the impact of incomplete or truncated lists on the quality of the resultant matching.

In this paper, we focus on the Serial Dictatorship mechanism in a model where each agent of the proposing side (students) has a random preference list of length $d$, sampled independently and uniformly at random from $n$ schools, each of which has one seat. Our main result shows that if the students primarily care about being matched to any school of their list (as opposed to ending up unmatched), then all students in position $i\leq n$ will prefer markets with longer lists, when $n$ is large enough. Schools on the other hand will always prefer longer lists in our model. We moreover investigate the impact of $d$ on the rank of the school that a student gets matched to.

Our main result suggests that markets that are well-approximated by our hypothesis and where the demand of schools does not exceed supply should be designed with preference lists as long as reasonable, since longer lists would favor all agents.
\end{abstract}

\begin{document}

\begin{titlepage}

\maketitle
\setcounter{tocdepth}{2} 
\tableofcontents

\end{titlepage}

\section{Introduction}\label{sec:intro}

One of the central objectives of theoretical research in matching markets is to design mechanisms and rigorously establish whether they, along with their outcomes, satisfy desirable properties. Depending on the specific market, the focus may be on concepts such as strategy-proofness, one- or two-sided optimality, or equilibria. These and other properties are often viewed as justifications for applying such mechanisms in practice, under the expectation that the valuable theoretical features will persist in the real world. However, in many applications, the assumptions that are necessary for such properties to hold may not be satisfied. For example, while in a static two-sided market one can compute a stable assignment efficiently, in real-world markets arising in public school allocation, a substantial number of agents often enter or leave the market after the first assignment has been decided (see for example~\cite{Irene}). The goal of the central planner becomes therefore less well-defined and, depending on the mathematical formalization, may lead to computationally hard problems, see for example~\cite{bampis2023online,faenza2024two,faenza2025minimum,miyazaki2019jointly}.

In this paper, we focus on one of the main restrictions encountered in real-world markets: imposing a limit on the maximum length of preference lists of agents. Such limits are common, for instance in school markets. While Gale and Shapley's Deferred Acceptance algorithm~\cite{gale1962college} assumes that students are allowed to list all schools they deem acceptable in their preference list, education departments traditionally impose a strict maximum limit on the number of schools an applicant may list\footnote{These include programs in Spain and Hungary~\cite{calsamiglia2010constrained}, Australia~\cite{artemov2017strategic}, as well as some cities in the United States.}. This restriction is motivated by practical concerns such as the additional burden that acquiring more information poses to schools and students alike. However, the effects of restricting the length of preference lists are significant on the properties of the mechanism and of its outcome. For instance, while the original Deferred Acceptance mechanism is strategy-proof for the proposing side, its implementation with bounded-size lists leaves room for strategic behavior of students. This is not merely a theoretical concern, but has practical implications~\cite{calsamiglia2010constrained} which are well known to central planners. For instance, until the school year 2024/2025, the Department of Education of New York City limited\footnote{This rule was changed for the school year 2025/2026, and applicants can now list as many schools as they want.} the number of schools in student preference lists to 12 and suggested that the applicants be strategic by reserving some of the slots for ``safe schools'' (schools where the student has high priority, and as a result, a high probability of being accepted)~\cite{NYC_DOE}. 

When deciding the length of preference lists in a matching market, it is therefore essential to carefully balance practical concerns with the potential efficiency losses that can result from limiting information exchange. A correct estimation of the efficiency loss induced by short(er) lists can therefore guide a market designer in striking the right trade-off.

\subsection{Our Contributions}\label{sec:our-contributions} In this paper, we focus on the impact that the length of preference lists has on the quality of the output matching in the Serial Dictatorship mechanism for two-sided matching markets. 
In our theoretical model, we assume that the proposing side (students) have preference lists drawn uniformly at random and that the disposing side (schools) have only one seat (we later discuss how these assumptions can be relaxed in our computational experiments). 

The concept of ``quality of a matching'' can be defined in multiple ways; we mostly focus on the probability that a student will be matched to any school of their preference list as opposed to remaining unmatched. This is justified by the fact that if a student cannot be matched to any school of their preference list, they completely lose control on the school they are matched to (in many markets, such students are assigned a remaining seat arbitrarily at the whim of the central planner). This seems by far the worse outcome for a student. We moreover study the probability that a student is matched to their top $k$ choices, a popular measure of quality of a matching\footnote{For instance, the National Resident Matching Program, while using the Deferred Acceptance algorithm, explicitly reports the number of applicants matched to their first choice~\cite{NRMP}, while the Boston Mechanism directly aims at maximizing the number of students matched to their first choice.}. 

\smallskip

\noindent {\bf 1.~Comparative analysis of balanced markets.} In our model, the probability $p_i$ that a student $i$ is matched to any school is determined by two opposite drivers. On one hand, the longer the preference lists, the more schools will have been matched to students with priority higher than $i$, decreasing $p_i$. On the other hand, longer preference lists imply that a student has a larger list of acceptable schools, leading to an increase in $p_i$.
Our main result is that under the Serial Dictatorship mechanism, when the number of schools $n$ is large enough, the latter driver dominates for all students in positions $i\leq n$: that is, $p_i$ increases with the length of lists. This is our Theorem~\ref{thm:main-discrete}. 

In particular, if the demand of schools does not exceed supply, all students will prefer longer lists. It is not hard to see that in such a market, all schools are matched with higher probability when preference lists are longer (see Lemma~\ref{lem:school-love}). Therefore our results suggest that markets that are well-approximated by these hypotheses should be designed with preference lists as long as reasonable, since longer lists favor all agents. 

\noindent {\bf 2.~Comparative analysis of general markets.} Interestingly, it is not true in general that longer lists will increase the probability of being matched to some school for every student. Indeed, for $n$ large enough, all students in position $i\geq \ceil{1.22n}$ will be matched with a higher probability when the length of preference lists is $1$ (students randomly choose one school to apply to), as opposed to when the length of preference lists is $2$, see Section~\ref{sec:continuous-to-discrete} and Figure~\ref{fig:d1_vs_d2}. More generally, for all $d<\ell$ and $n$ large enough, there exists some $i>n$ such that student $i$ is matched with higher probability when preference lists are of length $d$, as opposed to when preference lists are of length $\ell$, see Lemma~\ref{lem:crossin-discrete}.

\smallskip 

\noindent {\bf 3.~Absolute bounds on the probability of being matched.}
While the previous results compare the probabilities of a given student getting matched to any school between markets with different list length, they do not give us any information on the absolute probability. We show in Theorem~\ref{thm:bound-discrete} that as the length of preference lists increases, the probability that the student in position $n$ is matched (recall that $n$ is the number of schools) quickly approaches $1/2$. This is therefore also a lower bound (resp., an upper bound) to the probability of a student $i\leq n$ (resp., $i\geq n$) of being matched. In Lemma~\ref{lem:worst-case-rank}, we further prove bounds on the probability that a student in position $i\leq n$ gets matched to one of their top-$k$ choices as a function of the length of the market.

\smallskip

\noindent {\bf 4.~Numerical results.} In Section~\ref{sec:numerical}, we numerically study two extensions of the model.
We first focus on the case when the preference lists are not sampled uniformly at random, but rather, schools are sampled i.i.d. from one of $5$ distributions that place different weights on different schools. With the exception of one such distribution (a ``degenerate'' one), numerical experiments confirm our main result: that every student in position $i\leq n$ will be matched with higher probability when lists are longer. We then focus on the case when schools have $q>1$ available seats each, and confirm via simulations, that in this case too, students in balanced markets always prefer longer lists.

\subsection{Organization of the Paper}

We conclude this introductory section with further pointers to the literature. Section~\ref{sec:models} is devoted to introducing the main models and ideas, and formally stating many results (including those discussed above), without proofs. In particular, in Section~\ref{sec:discrete}, we introduce the main (discrete) model that we investigate in this paper. To analyze its properties for $n$ large enough, it will be useful to consider a continuous model that can be interpreted as a limit of the discrete model. This continuous model is introduced in Section~\ref{sec:continuous}. In Section~\ref{sec:continuous-to-discrete} we state the connections between the two models, as well as some relevant properties of the continuous model. We moreover discuss some implications of our results for Random Serial Dictatorship and extend our models to the case of schools having multiple seats in Section~\ref{sec:implications-extensions}. Last, we discuss the relevance of our model and our hypothesis in Section~\ref{sec:discussion}.

Proofs of results stated in Section~\ref{sec:models} appear in Section~\ref{sec:proofs}: we start with a limit theorem rigorously establishing the connection between the continuous and the discrete models (Section~\ref{sec:pf-cts-to-discrete}), followed by proofs of properties of the continuous market (Section~\ref{sec:pf-main-result}). Results from Section~\ref{sec:pf-cts-to-discrete} and Section~\ref{sec:pf-main-result} will allow to deduce properties of the discrete market (Section~\ref{sec:proof:discrete-market-probability} and Section~\ref{sec:proofs-impact-on-rank}). 

Numerical experiments testing the validity of our theoretical results when some assumptions are relaxed appear in Section~\ref{sec:numerical}. We conclude in Section~\ref{sec:conclusions}.

\subsection{Related Literature}

Random models for matching markets have a long and rich history. The model that is closest to ours is the one by~\cite{immorlica2015incentives}, where the authors consider random preference lists and assume that one of the two sides has preference lists of bounded length, providing an investigation of the model for large markets. In particular, \cite{immorlica2015incentives} show that with high probability and market size large enough, the core (i.e., the set of stable matchings) has small size. \cite{kanoria2021matching} study a similar model, investigating the change in the quality of the matchings in the core as a function of the preference list length. Similar questions have also been answered in the case when preference lists are complete and markets are balanced~\cite{pittel1989average,pittel1992likely} or unbalanced~\cite{ashlagi2017unbalanced,pittel2019likely}. 

The effect of short lists on the behavior of agents and on the quality of the output matching has been the subject of experimental studies~\cite{artemov2017strategic,bo2020iterative,calsamiglia2010constrained}. While specifics differ, they all share the common message that a loss of efficiency is experienced when agents are asked to report preference lists of bounded length, as opposed to preference lists of arbitrary length. 

Many theoretical studies have been devoted to Serial Dictatorship and Random Serial Dictatorship, with goals different from ours -- see, for instance,~\cite{abdulkadirouglu1998random,abdulkadirouglu2003school,bade2020random,bampis2023online,bogomolnaia2001new} and the references within the survey article~\cite{ehlers2021normative}.

To investigate our discrete model, we introduce a continuous model that can be thought of as a limit of the former. Continuous or semi-continuous models for two-sided markets have recently received quite some attention, since they often allow for tighter analysis, see, e.g.,~\cite{arnosti2022continuum,arnosti2023lottery,azevedo2016supply}. In particular, the work in~\cite{arnosti2023lottery} is closely aligned in its methods with our work. They similarly introduce a description of the limit of a discrete market via a solution to an initial value problem; then utilize Markov chain methods via an analogue of our Theorem~\ref{thm:xts} to show convergence of certain quantities relating to that market. Their differential equation is very general and allows for arbitrary distributions and varying lengths of preferences across students. The price one pays for such generality is tractability, and they do not produce expressions ripe for analytical solutions.
In contrast, our main technical contribution is to carefully construct an initial value problem for our specific case but that produces a simple interpretation connected to the preference of students and can be readily analyzed to yield insights for the original discrete model.

The outcome of the Serial Dictatorship mechanism in our model can also be understood as the application of a randomized version of the greedy algorithm on an online bipartite matching problem where, as usual, nodes from one side of the graph are given, while the others arrive one at the time, have degree exactly $d$, and are matched to one of their currently unmatched neighbors uniformly at random (or discarded if no such neighbor exists). Starting from the seminal work by~\cite{karp1990optimal}, many versions of online matching problems in bipartite graphs have been studied (see~\cite{devanur2022online,huang2023applications} for recent surveys). To the best of our knowledge, such models mostly focus on the competitive ratio of global objective functions (such as the number of matches or some profit or cost function associated to the matching), while our node-by-node analysis of the probability of being matched appears to be new. 

\subsection{Notation}

We denote by $\bbN=\set{1,2,3,\dots}$ and, for $k \in \bbN$,  $[k]=\set{1,2,\dots,k}$. We use $n\inN$ for the number of schools, $d\inN$ and $\ell\inN$ for lengths of preference lists, and $m\inN$ as the number of students (in cases where we do not assume infinite students). Superscripts denote properties that are fixed for the market (except in the continuous realm where we place the $d$ of $x_d(t)$ in the subscript for convenience), and subscripts are used for running indices. We use $i\inN$ for students, $j\in[n]$ for schools. Random variables are denoted by uppercase letters.

\section{Models and Results}\label{sec:models}

\subsection{A Discrete Random Market Model}\label{sec:discrete}

\paragraph{Model description.} Consider a random market consisting of $n \in \bbN$ schools and infinitely many students indexed by $i=1,2,3,\dots$, where each school has exactly one seat. Each student has a strict \emph{preference list}, consisting of $d\leq n$ schools, chosen uniformly at random from the set of $n$ schools. We call $d$ the \emph{preference list length}. If $a$ occurs before $b$ in this preference list, the student strictly prefers being matched to $a$ rather than $b$. If a school $c$ does not appear in a student's preference list, the student finds this school unacceptable and prefers remaining unmatched to being matched to that school. A school that appears in a student's preference list is \emph{acceptable} to the student.

Students are matched to schools via a \emph{Serial Dictatorship} mechanism as follows. Students are listed in some deterministic order; and each student in their turn picks their most preferred school that has an available seat. If the student finds none of the remaining schools acceptable, then that student goes unmatched. For a student $i\in[m]$, we denote by $K^{n,d}_i$ the random variable denoting the position of the school they get matched to within their preference list, setting $K^{n,d}_i=\infty$ if the student is unmatched (so $K^{n,d}_i\in\set{\infty,1,\dots,d}$). We define the random vector $\mcK=\set{K^{n,d}_1,K^{n,d}_2,K^{n,d}_3,\dots}$ as the realization of one run of the mechanism.

A student $i$ wishes to minimize $K^{n,d}_i$. That is, they want to be ranked to the most preferred school in their list. On the other hand, not getting matched at all ($K^{n,d}_i=\infty$) is by far worse than getting matched to a school that the student finds acceptable but ranks lower. Define $M^{n,d}_i=\indc_{\set{K^{n,d}_i<\infty}}$ to be the indicator random variable for the event that student $i$ gets matched at all. The probability that they get matched to any school is then $\mbb{M^{n,d}_i=1}$.

For a given market with a fixed number of schools, only $d$ is chosen by a centralized planner. It is therefore natural to ask what impact varying $d$ has on the distributional properties of $\mcK$, and on the distribution of individual students' ranks ($K^{n,d}_i$), as well as their probabilities of getting matched to any school, $\mbb{M^{n,d}_i=1}$.

\paragraph{Main results.} As our main result\footnote{All proofs of the results from this subsection appear in Section~\ref{sec:proof:discrete-market-probability} and Section~\ref{sec:proofs-impact-on-rank}.},  we show that for a number $n$ of schools large enough, every student in position $i\leq n$ will be matched to a school with higher probability in markets with longer lists. Formally, we have the following result.

\begin{restatable}{theorem}{thmMainDiscrete}\label{thm:main-discrete}
Let $d,\ell\inN$ with $\ell\geq d$. For every $n$ large enough and $i\leq n$, we have
\begin{align}
\mbb{M_i^{n,\ell}=1}\geq\mbb{M_i^{n,d}=1}.\label{eq:main-probability}
\end{align}
\end{restatable}

As we show in the following lemma, for each school, the probability of being matched increases with the length of preference lists. For a school $j\in [n]$, define $H^{n,d}_j(i)$ to be the indicator random variable for the event that school $j$ gets matched to some student just before student $i$ has had their turn. Then the following holds. 

\begin{restatable}{lemma}{thmLemmaSchoolsLike}\label{lem:school-love}
Let $d,\ell,n\inN$ with $d\leq \ell\leq n$. For every $j \in [n]$ and every $i=1,2,\dots$, we have: 
\begin{align*}
\mbb{H^{n,\ell}_j(i)=1}\geq\mbb{H^{n,d}_j(i)=1}.
\end{align*}
\end{restatable}

Theorem~\ref{thm:main-discrete} and Lemma~\ref{lem:school-love} imply that if the number of students does not exceed the number of schools, for $n$ large enough all agents will be matched with higher probability in the market where preference lists are longer. Interestingly, Theorem~\ref{thm:main-discrete} does not necessarily hold for $i$ larger than $n$: as we discuss in Section~\ref{sec:continuous-to-discrete}, for any $n$ large enough, students in position $i=\lceil 1.22n \rceil$ and beyond have a higher probability of being matched when preference lists are of length $1$ as opposed to length $2$. A similar phenomenon happens also for longer preference lists, as formalized by the next lemma.

\begin{restatable}{lemma}{thmCrossingDiscrete}\label{lem:crossin-discrete}
Let $d,\ell\inN$ with $\ell\geq d$. For every $n$ large enough, there exists $i>n$, such that
\begin{align*}
\mbb{M_i^{n,\ell}=1}\leq\mbb{M_i^{n,d}=1}.
\end{align*}
\end{restatable}

While Theorem~\ref{thm:main-discrete} allows us to compare probabilities of getting any school for students under two different list lengths, it does not give us any information on the absolute probability, $\mathbb{P}(M_i^{n,d}=1)$, that a student $i$ will get matched. This probability decreases with $i$ and eventually converges to $0$ as $i$ becomes large enough (since all schools will have been taken), so a natural question is to ask how this probability behaves for smaller values of $i$. It turns out that as $d$ increases, $\mathbb{P}(M_n^{n,d}=1)$ quickly approaches $1/2$.  This is significant for all students $i\leq n$ (since $\mathbb{P}(M_i^{n,d}=1)$ decreases in $i$) and for all students when the number of students is comparable to the number of schools.

\begin{restatable}{theorem}{thmBoundDiscrete}\label{thm:bound-discrete}
Let $d \in \mathbb{N}$. For every $n$ large enough, we have 
\begin{equation}\label{eq:bound-on-P-Mid=1}
\frac{d}{2d+1} \leq \mbb{M_n^{n,d}=1} \leq \frac{2d}{4d+1}.
\end{equation}
In particular,
\begin{align*}
\lim_{d\to\infty} \lim_{n\to\infty} \mathbb{P}(M_n^{n,d} =1)=\frac{1}{2}.
\end{align*}
\end{restatable}

We additionally prove the following lemma that bounds the change in probability of a given student getting matched to one of their top-$k$ schools for $k\leq d$. 

\begin{restatable}{lemma}{lemWorstCaseRank}\label{lem:worst-case-rank}
Let $d\in\bbN$. For every $n$ large enough, $k\leq d$ and $i\leq n$, we have
\begin{align*}
\mbb{K^{n,d}_i\leq k}-\mbb{K^{n,d+1}_i\leq k}&\leq\br{\frac{d+2}{2d+3}}^{k/(d+1)}-\br{\frac{2d+1}{4d+1}}^{k/d}.
\end{align*}
\end{restatable}

\paragraph{The dynamics of the discrete model.} To analyze the discrete model, it is useful to study the number of students matched to any school just before it is student $i$'s turn, given by
\begin{align*}
T^{n,d}_i&=\sum_{j=1}^{i-1}\indc_{\set{K^{n,d}_j<\infty}}.
\end{align*}
Since each school has exactly one seat, $T^{n,d}_i$ coincides with the number of schools matched to the students $\set{1,\dots,i-1}$. The number of schools that remain unmatched before student $i$'s turn is therefore $n-T^{n,d}_i$.

We remark on one property of the random market, called the \emph{principle of deferred decisions}. This has proved useful in analyzing other markets, see, e.g.,~\cite{immorlica2015incentives}. Observe that the preference list of student $i$ does not play a role in the output until the moment the $i$-th round of Serial Dictatorship, when student $i$ is selected to pick their favorite school. We therefore do not need to specify $i$'s preference list until this moment, allowing us to defer this decision until the list is needed. In particular, note that the distribution of $K^{n,d}_i$ depends only on $T^{n,d}_i$: we only need to know the number of schools that remain unmatched at the start of the turn to know the distribution of the position that the student gets matched to. This is further illustrated below.

Suppose it is the $i$-th student's turn, and the students prior to $i$ have been matched to $k$ schools, so $T^{n,d}_i=k$. It is then straightforward to compute the probability that a list of $d$ schools chosen uniformly at random from the set of $n$ schools will overlap with any of the unmatched $n-k$ schools. This probability is given by the following formula.
\begin{align}\label{eq:discrete-match-prob}
\mbb{M^{n,d}_i=1\given T^{n,d}_i=k}=\piecewise{1-\frac{\binom{k}{d}}{\binom{n}{d}}, & k\geq d, \\ 1, & \text{otherwise}.}
\end{align}
It is immediately clear from this formula that the probability of getting matched depends only on $n$, $d$, and $T^{n,d}_i$. Furthermore, all else held constant, $\mbb{M^{n,d}_i=1\given T^{n,d}_i=k}$ decreases as $T^{n,d}_i$ increases. The principle of deferred decisions allows us to connect our discrete model to a continuous one, discussed next. 

\subsection{A Continuous Market Model}\label{sec:continuous}

We now introduce a continuous market model, which we will show is equivalent, to the limit (uniformly in probability) of the previously introduced discrete market model in the next section.

\paragraph{Description.} Consider a (deterministic) continuous matching model consisting of a unit mass of schools and a continuum of students represented by the interval $[0,\infty)$ and indexed by $t\geq 0$. Let $d\geq 1$ be a parameter of the market, and define $x_d(t)$ to be the proportion of schools matched to students in $[0,t)$. The market is now defined by the Initial Value Problem
\begin{equation}\label{eq:ivp}
x_d(0) = 0, \quad x_d'(t) = 1 - x_d(t)^d.
\end{equation}

The interpretation of the Initial Value Problem~\eqref{eq:ivp} is as follows: start with $x_d(0)=0$ and traverse students in order from $t=0$. At the time of student $t$, they get matched to schools at rate $1-x_d(t)^d$, which is also the rate by which $x_d(t)$ increases at $t$. We note $x_d'(t)$ depends only on the proportion of schools matched until that point, and that $x_d(t)$ is smooth in its argument and takes values in $[0,1]$.

\paragraph{Dynamics.} The dynamics of the continuous model are described by the (conceptually) simple initial value problem in~\eqref{eq:ivp}. Note that $x_d$ is everywhere increasing, and $x_d(t)\to 1$ as $t\to\infty$, but $x_d'(t)\to 0$ as $n\to\infty$. Further, one immediately observes that for $t\approx 0$, $x_d'(t)\approx 1$: at first all students are getting matched at the highest possible rate. The parameter $d$ is a measure of intensity of matching: if the proportion of schools matched were fixed, smaller $d$ would yield a slower rate of matching.

\begin{figure}[!h]
\centering
\includegraphics[width=.8\linewidth]{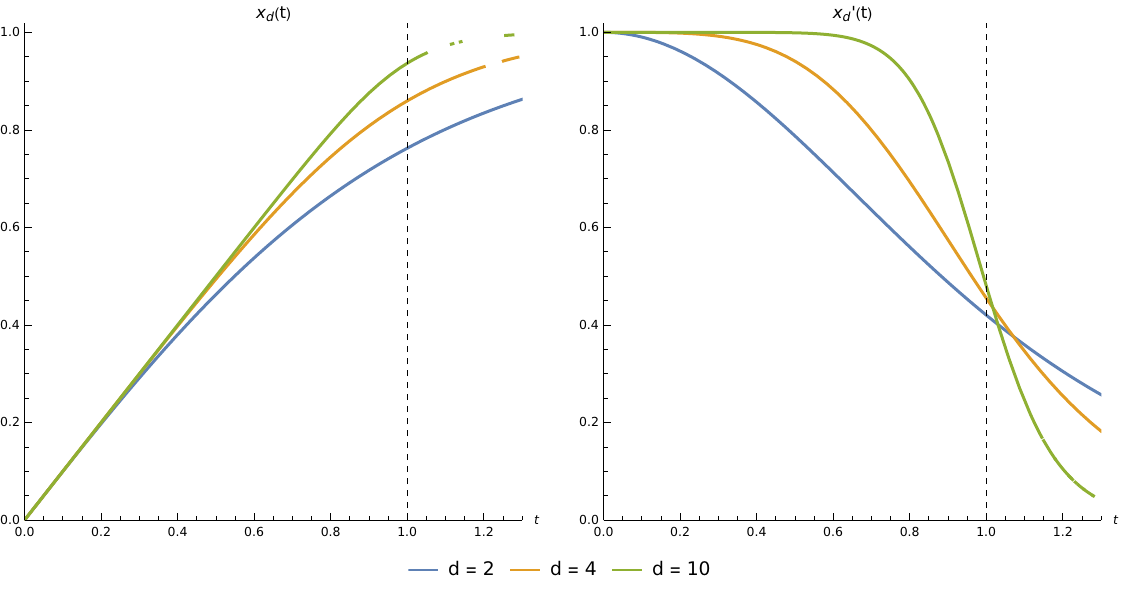}
\label{fig:xd_and_xd_prime}
\caption{$x_d$ and $x_d'$ for various values of $d$: note how the values for $x_d'$ cross just after $t=1$} 
\end{figure}

Similarly to the discrete model (see the discussion in Section~\ref{sec:our-contributions}), we see that there are again two forces at play: for small $d$, the rate of getting matched is smaller, but simultaneously the proportion of schools taken up by earlier students is smaller, so later students may prefer small $d$ as it gives them some schools to possibly get matched to. Indeed, for $t$ exceeding $1$, the probability of getting matched quickly vanishes as $d$ increases. On the other hand, a large $d$ increases the rate of getting matched, but also reduces the number of schools remaining by the time of student $t$'s turn.

We remark that although the description of the initial value problem might look deceptively straightforward, actually analyzing it proves to be quite challenging. The differential equation is highly non-linear, a special case of the \emph{Chini type} differential equation studied at least since 1924~\cite{kamkeDE}, to which there is no known analytical solution for general $d$ \cite{henk2022series}. To circumvent this, we must find implicit ways of proving properties about the solutions of the differential equation that do not require a full analytical solution.

\subsection{From Continuous to Discrete}\label{sec:continuous-to-discrete}

In this section we discuss the main theorem that connects the continuous and discrete markets, then state our main results within the continuous realm, and then describe how they carry over to the discrete case.

\paragraph{Connection to the discrete model.} We now give an intuition on how the discrete market with fixed list length $d$ approaches the continuous market when the number of schools $n\to\infty$. We make this limit rigorous via a functional law of large numbers in Theorem~\ref{thm:xts}.

To see the connection, consider the discrete market described in Section~\ref{sec:discrete} for a fixed list length $d$, and suppose we are at the turn of student $t=i/n$ for some fixed $n$. In Lemma~\ref{lem:match-prob-approx}, we show that
\begin{align}\label{eq:approx-match-prob}
\mbb{M^{n,d}_i=1}=1-(T^{n,d}_i/n)^d+O(d^2/n).
\end{align}
Here $T^{n,d}_i/n$ is the proportion of schools taken by students $\set{1,\dots,i-1}$ and so is an analogue of $x_d(t)$ (the proportion of schools taken by students $[0,t)$) in the continuous model. In the limit as $n\to\infty$, the last term vanishes so when student $t=i/n$ has their turn, they get matched to a school with probability $x_d'(t)=1-x_d(t)^d$. Indeed, we show in Lemma~\ref{lem:prob-to-xprime} that $\mbb{M^{n,d}_i=1}\to x_d'(i/n)$ in probability as $n\to\infty$. In the continuous model $d$ is therefore the analogue of the list length (but is now allowed to be any real number in $[1,\infty)$).

\begin{figure}[!h]
\centering
\includegraphics[width=0.8\linewidth]{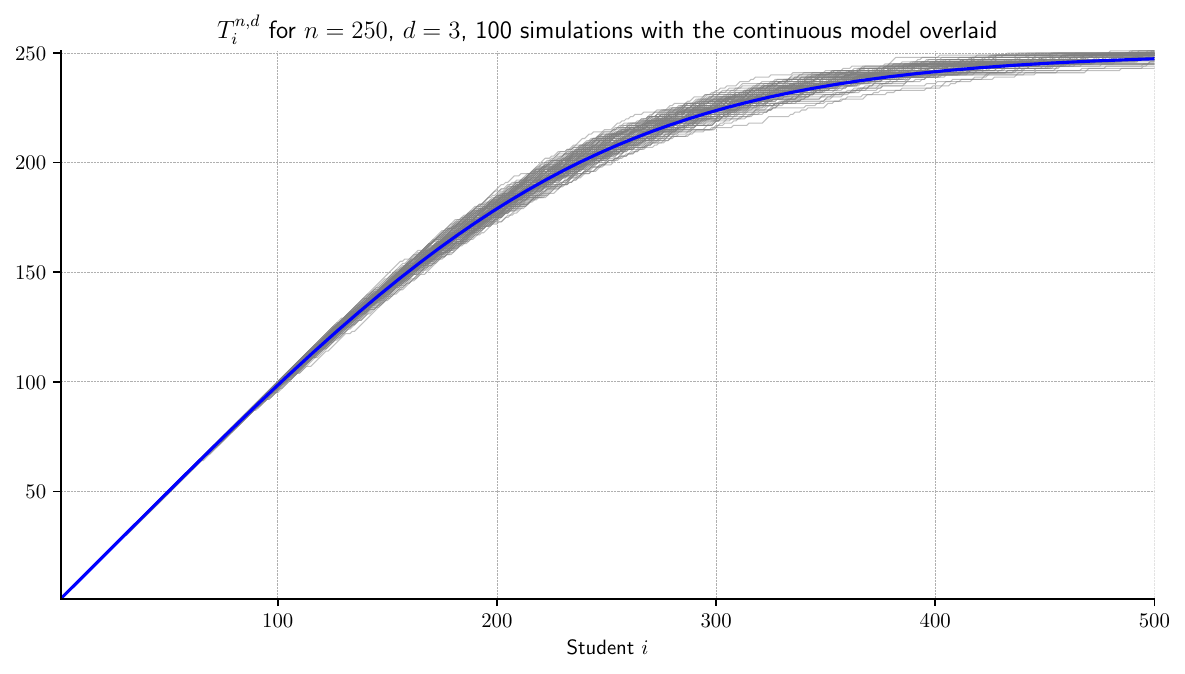}
\caption{100 simulations of the discrete model with the continuous model overlaid.}
\end{figure}

The following theorem, proved in Section~\ref{sec:pf-cts-to-discrete}, describes rigorously the sense in which the continuous market is a limit of the discrete market. 

\begin{restatable}{theorem}{thmXts}\label{thm:xts}
 Fix $d\in\bbN$ and define $p(x)=1-x^d$ for $x\in[0,1]$. For $t\geq 0$, let $T^{n,d}_{\floor{tn}}$ be the random variable denoting the number of schools matched by the first $\floor{tn}$ students when there are $n$ schools in the discrete random market  with list length $d$. Then $n^{-1}T^{n,d}_{\floor{tn}}\to x_d(t)$ uniformly in probability as $n\to\infty$, where $x_d(t)$ is the unique solution satisfying $x_d(0)=0$ and $x_d'(t)=p(x_d(t))$ for $t\geq 0$. That is, for all $s\geq 0$ and $\eps>0$, as $n\to\infty$,
\begin{align}\label{eq:converges-uniformly}
\mbb{\sup_{t\in[0,s]}\abs{n^{-1}T^{n,d}_{\floor{tn}}-x_d(t)}\geq\eps}\to 0.
\end{align}
Moreover, as $n\to\infty$ for all $r\geq1$, we have
\begin{align}\label{eq:mean-converges}
\meb{\abs{n^{-1}T^{n,d}_{\floor{tn}}-x_d(t)}^r}\to 0.
\end{align}
\end{restatable}

In particular,~\eqref{eq:converges-uniformly} implies that as $n\to\infty$, the number of schools taken prior to student $t=i/n$ converges in probability to the solution of the Initial Value Problem~\eqref{eq:ivp}, i.e., the continuous market. Not only does the proportion of schools taken converge in probability for every point $t\geq 0$, but the maximum deviations of the random market around the continuous solution vanish in probability, and the mean converges to the continuous solution (by taking $r=1$ in~\eqref{eq:mean-converges}). The theorem additionally posits that the initial value problem defining the continuous market, that is \eqref{eq:ivp}, has a unique solution satisfying the boundary data and that this solution extends for all time $t\geq 0$. 

The next lemma rigorously connects the match rates of the discrete and continuous models in this limit. Recall that $M^{n,d}_i$ is the indicator random variable for whether student $i$ gets matched to any school in the discrete random market with $n$ schools and preference lists of length $d$.

\begin{restatable}{lemma}{lemProbToXPrime}\label{lem:prob-to-xprime}
For all $d\inN$, $t\geq 0$, we have $\mbb{M^{n,d}_{\floor{tn}}=1}\to x_d'(t)$ as $n\to\infty$.
\end{restatable}

To understand whether $\mbb{M^{n,d}_i=1}$ is larger in longer or shorter lists, we therefore need to understand whether $x_d'(t)$  is larger for larger or smaller $d$. That is, if $x_{d+1}'(t)\geq x_d'(t)$ for some $t\geq 0$, then student $t=i/n$ prefers lists of length $d+1$ to those of length $d$ for every large enough $n$. If we can show such a relationship holds for every $d\geq 1$ at $t$, then clearly student $t$ always prefers longer lists to shorter ones. Our task of understanding match probability therefore becomes understanding $x_d'(t)$ in the continuous market. 

\paragraph{The case of $d=1$ and $d=2$.} For $d=1$ and $d=2$, we can in fact compute analytic solutions to the initial value problem~\eqref{eq:ivp} defining the continuous market. For the case of $d=1$, we unsurprisingly get
\begin{align*}
x_1(t)=1-e^{-t},\qquad x_1'(t)=e^{-t},
\end{align*}
and for $d=2$ we get
\begin{align*}
x_2(t)=\frac{e^t-e^{-t}}{e^t+e^{-t}};\qquad x_2'(t)=\frac{4e^{2t}}{(e^{2t}+1)^2}.
\end{align*}
The latter may also be written $x_2(t)=\tanh(t)$ and $x_2'(t)=\sech(t)^2$ in terms of the hyperbolic functions. With these analytic expressions, one can solve to find $x_2'(t)\geq x_1'(t)$ if and only if $t\leq 1.219$. This means that for students before this cutoff, they prefer lists of length $2$ to lists of length $1$, and vice versa for larger $t$, students prefer the case of shorter lists: see Figure~\ref{fig:d1_vs_d2}. We wish to highlight two key phenomena from this example: all students for $t\in[0,1]$ have a higher probability of getting matched to any school with longer lists compared to shorter lists, and there is a cutoff (after $t=1$) where this pattern reverses. 

\begin{figure}[!h]
\centering
\includegraphics[width=.8\textwidth]{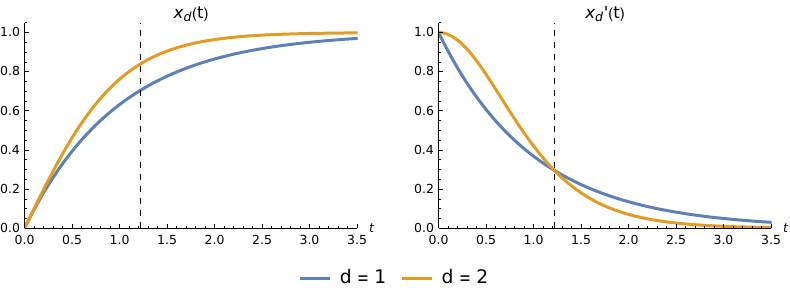}
\caption{$x_d$ and $x_d'$ for $d=1$ and $d=2$. Note that students before the highlighted cutoff at $t=1.219$ have higher probability of being assigned to any school with longer lists, and vice versa for students after the cutoff.}\label{fig:d1_vs_d2}
\end{figure}

\paragraph{The case of general $d$.} In the case of general $d$, we cannot hope to solve the initial value problem analytically. Our main result in the continuous realm, proved in Section~\ref{sec:pf-main-result} is that all students for $t\leq 1$ always prefer longer lists.

\begin{restatable}{theorem}{thmBoundCts}\label{thm:bound-cts}
For all $d,\ell\in[1,\infty)$ with $\ell\geq d$, and for all $t\in[0,1]$, we have $x_\ell'(t)\geq x_d'(t)$ .
\end{restatable}

\subsection{Implications and Extensions}\label{sec:implications-extensions}

\paragraph{Random Serial Dictatorship.} In the Random Serial Dictatorship mechanism, the order in which students propose is sampled from all possible permutations of students. When applied to our setting, it introduces an additional layer of randomness, on top of the one due to the preferences of agents. 

Formally, consider a modification of our model, where there are $n$ schools (each with one seat), $m\in \mathbb{N}$ students ordered randomly, and each of them has a strict preference list of $d$ schools, chosen uniformly at random from the set of $n$ schools. For $d,n \in \mathbb{N}$, we denote by $R^{n,d}$ the indicator random variable for the event that a student is matched to some school under the Serial Dictatorship algorithm, where students propose following the sampled order (note that this probability does not depend on the specific student, while it depends on $m$, but we are omitting this dependency for the sake of brevity). The following is an easy corollary of Lemma~\ref{lem:school-love}; its proof is given in Appendix~\ref{appx:lemmas_and_proofs}. 

\begin{restatable}{corollary}{corSerial}\label{cor:Serial}
Let $d,\ell,n\inN$ with $n\geq \ell\geq d$. We have
\begin{align*}
\mbb{R^{n,\ell}=1}\geq\mbb{R^{n,d}=1}.
\end{align*}
\end{restatable}

\paragraph{Schools with multiple seats.}\label{para:multiple-seats}

In this section we discuss an extension of the discrete random market introduced in Section~\ref{sec:discrete} to the case where schools each have $q\inN$ seats. Like in the simpler market, we are interested in students $i=1,2,\dots$ as they each have their turn and are (possibly) matched to one of their remaining acceptable schools. Formally, denote by $M^{n,d,q}_{i}$ the indicator random variable of whether student $i$ gets matched with $n$ schools, $d$ length lists, and $q$ seats per school. Based on numerical experiments (some of which are reported in Section~\ref{sec:conclusions}), we conjecture that the following extension of Theorem~\ref{thm:main-discrete} holds.

\begin{conjecture}\label{conj:q_ge_1}
Let $d,\ell\inN$ with $\ell\geq d$. For every $n$ large enough and $i\leq nq$, we have
\begin{align*}
\mbb{M_i^{n,\ell,q}=1}\geq\mbb{M_i^{n,d,q}=1}.
\end{align*}
\end{conjecture}

We now sketch some ideas on how the techniques introduced in this paper could be used to attack Conjecture~\ref{conj:q_ge_1}. More details are given in Appendix~\ref{appx:q_ge_1}.

We fix $q\inN$ and do not include it in superscripts for brevity. Additionally, fix $n\inN$ as the number of schools, and $d\geq 1$ as the list length. It is now not enough to keep track just of $T^{n,d}_i$, now defined to be the number of schools with all seats taken after students $\set{1,2,\dots,i-1}$ have had their turn; instead, we must additionally keep track of the vector $S_i=(S_i^0,S_i^1,S_i^2,\dots,S_i^{q-1})$ counting the number of schools with $k=0,1,\dots,q-1$ seats taken at this point. (Note that $T^{n,d}_i$ is exactly the number of schools with $q$ seats taken). Given this information, we can now see the market dynamics. When a student $i$ has their turn, they pick $d$ schools uniformly at random from the set of all schools. The probability that the student picks in this list any school that has any seats remaining (and thus the student gets matched) is still the same one given in~\eqref{eq:discrete-match-prob}, and is approximated by the same expression as Lemma~\ref{lem:match-prob-approx}. If the student picks a school with a remaining seat, it will have $k$ seats taken with probability $S_i^k/\sum_{j=0}^{q-1}S_i^j$. The update rule for the counts is then as follows. If the school has, say $k$ seats taken we then have $S_{i+1}^k=S_i^k-1$, and if $k<q-1$ then $S_{i+1}^{k+1}=S_i^{k+1}+1$ and if $k=q-1$ then $T^{n,d}_{i+1}=T^{n,d}_i+1$.

In the case of $q=1$ we observed interesting behavior near $i\approx n$. Then with $q>1$ we would look for this at $i\approx nq$ due to the higher number of seats.

In Section~\ref{sec:continuous} we introduced a continuous deterministic model with $q=1$ and in Section~\ref{sec:continuous-to-discrete} we showed how the discrete random model approaches the continuous model in probability as $n\to\infty$. Similarly, it is possible to construct a continuous model for arbitrary $q=1,2,3,\dots$ and prove a similar limit result. 
We give an outline of this process in Appendix~\ref{appx:q_ge_1} and discuss strategies for solving the resultant differential equations. The missing piece to prove Conjecture~\ref{conj:q_ge_1} is an analogous of Theorem~\ref{thm:bound-cts}. Again, we refer to Appendix~\ref{appx:q_ge_1} for details.

\subsection{Discussion of the Discrete Model}\label{sec:discussion}

Models where agents have preference lists that are sampled uniformly at random from the list of all permutations (possibly, as in our paper, of bounded size) form a very common theoretical assumption for the study of two-sided matching markets, see for example~\cite{ashlagi2017unbalanced,immorlica2015incentives,pittel1989average,pittel1992likely,pittel2019likely}. One of the reason for the popularity of these models is that they often lead to tractable expressions. Results can then be verified to hold numerically for models with correlated preferences, or this can motivate further analytical study of such more cumbersome markets. analytically. 
Our approach follows this general line of inquiry, with our main result (Theorem~\ref{thm:main-discrete}) proved for the case where lists are sampled uniformly at random, and verified numerically in Section~\ref{sec:numerical} for preference lists sampled from more complex distributions.

The choice of Serial Dictatorship for the mechanism investigated in this paper is motivated by multiple facts. First, its relevance: it is broadly studied in theory (including its randomized version, see, e.g.~\cite{bampis2023online,bogomolnaia2001new}) and applied in practice to many settings, most famously house allocation problems~\cite{abdulkadirouglu1999house}. A Serial Dictatorship mechanism also naturally arises in the case of two-sided stable matching problems where schools have a shared (homogeneous) and strict preference list over the students which is, for instance, dictated by test scores, or by a lottery system where the lottery number determines priority~\cite{arnosti2023lottery}. In such cases, the Deferred Acceptance mechanism reduces to Serial Dictatorship, with students proceeding according to the shared order and picking their favorite remaining school. Both of these scenarios are common in school districts. For instance, in New York primary school matching the Department of Education has very limited information on students and so assigns priorities largely at random, and completely at random in cases such as city-wide gifted and talented schools (where the admissible students are restricted to those who are considered gifted or talented). On the other hand, the NYC specialized high schools are by law required to rank students uniquely by their score on the standardized Specialized High School Admissions Test (SHSAT) only~\cite{faenza2023discovering}.

Another reason to study the Serial Dictatorship mechanism is that it endows students with an obvious order, allowing granular statements to be made about each student uniquely (such as in our case discussing outcomes based on order). In many other mechanisms where randomness comes into play, students become indistinguishable, and one can hope for at most aggregate level statements and results.

\section{Proofs of Main Results}\label{sec:proofs}

\subsection{Connections Between the Discrete and Continuous Markets}\label{sec:pf-cts-to-discrete}

In this section we prove the following theorem, which characterizes the behavior of the rate at which students get matched to schools as $n\to\infty$ via a solution to a differential equation.

\thmXts*

We begin by computing an approximation to the probability of a student getting matched in the discrete case. 

\begin{lemma}[Probability of getting matched]\label{lem:match-prob-approx}
In the discrete market with $n$ schools, we have
\begin{align*}
\abs{\mbb{M^{n,d}_i=1\given T^{n,d}_i=k}-\br{1-(k/n)^d}}\leq d^2/n.
\end{align*}
That is, when $k$ schools have been matched prior to student $i$, then the probability that student $i$ gets matched to any school is approximated by $1-(k/n)^d$.
\end{lemma}
\begin{proof}
Recall that when it is student $i$'s turn, they get matched to any school with probability given in \eqref{eq:discrete-match-prob}, which for $k\geq d$ states $\mbb{M^{n,d}_i=1\given T^{n,d}_i=k}=1-\binom{k}{d}/\binom{n}{d}$. This is the complement of the probability that a random list of $d$ schools sampled from $n$ without replacement will overlap entirely with the list of $k$ taken schools. Consider now an alternative procedure where student $i$ samples a list of $d$ schools \emph{with replacement}, where each school has a probability of $k/n$ of being already taken by previous students. The probability of sampling a list with at least one school not taken by previous students conditioned on exactly $k$ schools having been matched to students $1,\dots, i-1$ would be $1-(k/n)^d$.

It remains to bound the difference between sampling the preference lists with replacement and sampling them without replacement. Observe that the outcomes differ exactly in the case that we choose at least one pair of schools that are the same while sampling with replacement. Otherwise the outcomes coincide. Formally, consider the probability space generated by the student sampling with replacement, assuming $k$ schools are matched to students $1,\dots,i-1$. Let $E$ be the event that student $i$ gets matched to any school and let $X$ be the event that the sample contains a repeated school. We have $\mbb{M^{n,d}_i=1\given T^{n,d}_i=k}=\mbb{E\given X^c}$, and $\mbb{E}=1-(k/n)^d$. Now with the law of total probability, write
\begin{align*}
\abs{\mbb{E}-\mbb{E\given X^c}}
&=\abs{\mbb{E\given X}\mbb{X}+\mbb{E\given X^c}(1-\mbb{X})-\mbb{E\given X^c}} \\
&=\mbb{X}\abs{\mbb{E\given X}-\mbb{E\given X^c}} \\
&\leq\mbb{X}.
\end{align*}
Finally we bound $\mbb{X}$, the probability that a sample drawn with replacement contains duplicate items. Observe that for each of the $d(d-1)/2$ pairs of positions in the sample there is exactly a $1/n$ probability that that pair is the same. Using the union bound, this yields $\mbb{X}\leq\frac{d(d-1)}{2n}\leq d^2/n$.
\end{proof}

Our aim is to show that the discrete market converges to the continuous one as $n$ grows large. 
We assume here that reader is familiar with the theory of ODE and Markov processes. A gentle introduction to the relevant background can be found in Appendix~\ref{app:background-ODE-MT}.

To connect the behavior of the discrete market with that of the continuous market, we apply the following theorem. See \cite[Theorem~17.3.1]{kroese2013handbook} for this version and \cite{ethier2009markov} for a proof. 

\begin{theorem}[A Functional Law of Large Numbers]\label{thm:flln}
Suppose $\set{J^n_t}_{n=1,2,\dots}$ is a family of continuous time Markov chains on finite state spaces where for $t\geq 0$, $J^n_t$ takes values in $\mcS^n\subseteq\bbZ$ and has transition rate matrix $Q^n=(q^n(i,j);i,j\in\mcS^n)$. Suppose that there is a subset $\mcD\subseteq\bbR$ and a family $\set{f^n}_{n=1,2,\dots}$ of functions with $f^n:\mcD\cross\bbZ\to\bbR$ which are bounded and continuous in the first argument such that, for each $i\in\mcS^n$ and $k\in\mathbb{Z}\setminus\{0\}$ such that $i+k\in\mcS^n$, we have
\begin{align}\label{eq:qs-and-fs}
q^n(i,i+k)=nf^n(i/n,k).
\end{align}
Define $g^n(x)=\sum_{k\in\bbZ} kf^n(x,k)$ for $x\in\mcD$, and suppose there exists a Lipschitz continuous function $g:\mcD\to\bbR$ such that the $g^n$ converge uniformly to $g$ on $\mcD$. Suppose $\lim_{n\to\infty}n^{-1}J^n_0=x_0$ for some $x_0$. Then there exists a unique deterministic trajectory $x(t)$ satisfying $x(0)=x_0$ and $x'(t)=g(x(t))$, $x(t)\in\mcD$, $t\in[0,T]$, and $\{n^{-1}J^n_t\}_{t\geq 0}$ converges uniformly in probability on $[0,T]$ to $x(t)$. 
\end{theorem}

We are now ready to prove Theorem~\ref{thm:xts}.

\begin{proof}[Proof of Theorem \ref{thm:xts}]
As in the statement of the theorem, fix $d\geq 1$ and define $p(x)=1-x^d$ for $x\in[0,1]$. 
For $t \geq 0$, define $X^n_t=T^{n,d}_{\floor{nt}}$. Letting $\Delta=n^{-1}$ for convenience, note that $\{X^{n}_t\}_{t=0,\Delta,2\Delta,\dots}$ is a discrete time Markov chain on the state space $\mcS^n=\set{0,1,\dots,n}$ with transition probabilities
\begin{align*}
\mbb{X^n_{t+\Delta}=j\given X^n_t=i}&=\piecewise{p^n_i,& j=i+1,\\ 1-p^n_i,& j=i, \\ 0,& \text{otherwise}.}
\end{align*}
where $p^n_i=1-(i/n)^d\pm O(d^2/n)$ by Lemma~\ref{lem:match-prob-approx}.

Starting from $\{X_t^n\}_{t=0,\Delta,2\Delta}$, we will now construct a sequence of continuous time Markov chains that satisfy the conditions of Theorem~\ref{thm:flln}, then bring back the result to the discrete time chain.

Associate to each discrete time chain $\{X^n_t\}_{t=0,\Delta,2\Delta,\dots}$, a coupled continuous time Markov chain $\{J^n_t\}_{t\geq 0}$ as follows: let $H^n_t$ be a homogeneous Poisson process with rate $\lambda=n$, and let $\{J^n_t\}_{t\geq 0}$ be defined by $J^n_t=X^n_{\Delta H^n_t}$. This is a well-known embedding technique (see Theorem~\ref{thm:embedding} in Appendix~\ref{appx:mc} for a statement and discussion) and yields the transition rate matrix
\begin{align*}
q^n(i,j)&=\piecewise{n p^n_i,& j=i+1,\\ -n p^n_i,& j=i, \\ 0,& \text{otherwise}.}
\end{align*}
Set $f^n(x,1)=p^n_{xn}$ and $f^n(x,\cdot)=0$ otherwise. We now verify that we can apply Theorem~\ref{thm:flln}. First, let $\mcS^n=\{0,1,2,\dots,n\}$ and ${\mcD}=[0,1]$. Note $q^n,f^n$ satisfy~\eqref{eq:qs-and-fs}: we only need to check $k=1$ in~\eqref{eq:qs-and-fs}, which we verify as
\begin{align*}
q^n(i,i+1)&=np^n_i \\
&=nf^n(i/n,1).
\end{align*}
For each $n$, $f^n$ is clearly bounded and continuous in the first argument. We moreover have for all $x\in[0,1]$
\begin{align*}
g^n(x)&=\sum_{k\in\bbZ}kf^n(x,k) \\
&=p^n_{xn} \\
&=p(x)\pm O(d^2/n), & \hbox{(by Lemma~\ref{lem:match-prob-approx})}
\end{align*}
hence $g^n\to p$ uniformly, since $\abs{g^n(x)-p(x)}= O(d^2/n)$ is independent of $x$. Observe that $p$ is Lipschitz-continuous with Lipschitz-constant $d$. Since $J^n_0=0$, we have $n^{-1}J^n_0=0$.

We can therefore apply Theorem~\ref{thm:flln} and deduce that $\{n^{-1}J^n_t\}_{t \geq 0}$ converges uniformly in probability to the unique solution $x_d(t)$ of the initial value problem~\eqref{eq:ivp}.

To carry over convergence uniformly in probability to the discrete Markov chain $\{X_t^n\}_{t=0,\Delta, 2\Delta,\dots}$, fix some $s\geq 0$. Recall that, for $\tau \geq 0$, $H^n_\tau$ is the count of events in the underlying Poisson process up to time $\tau$. Since every event of $H^n_t$ corresponds to to either a unit jump in $\{X^n_t\}_{t=0,\Delta,2\Delta,\dots}$ (that is, $X^n_{t+\Delta}=X^n_{t}+1$) or no transition (that is, $X^n_{t+\Delta}=X^n_{t}$), we have for each $t \geq 0$.
\begin{align*}
\abs{X^n_t-J^n_t}&=\abs{X^n_t-X^n_{\Delta H^n_t}} \\
&\leq\abs{\floor{nt}-H^n_t} \\
&\leq1+\abs{nt-H^n_t}.
\end{align*}
This yields

\begin{align}
\sup_{t\in[0,s]}\abs{n^{-1}X^n_t-x_d(t)}&\leq\sup_{t\in[0,s]}\abs{n^{-1}X^n_t-n^{-1}J^n_t}+\sup_{t\in[0,s]}\abs{n^{-1}J^n_t-x_d(t)} \nonumber \\
&\leq\frac{1}{n}+\sup_{t\in[0,s]}\abs{t-n^{-1}H^n_t}+\sup_{t\in[0,s]}\abs{n^{-1}J^n_t-x_d(t)}.\label{eq:split-soup}
\end{align}

Note that $n^{-1}H^n_t-t$ is a martingale\footnote{For $s\geq t$, using the independence of increments property and the fact that $\meb{H^n_{s-t}}=n(s-t)$, we have
\begin{align*}
\meb{n^{-1}H^n_s-s\given H^n_t}-(n^{-1}H^n_t-t)=\meb{n^{-1}(H^n_s-H^n_t)\given H^n_t}-(s-t)=\meb{n^{-1}(H^n_{s-t})}-(s-t)=0.
\end{align*}}, and applying Doob's martingale inequality (see Theorem~\ref{thm:doob} of Appendix~\ref{appx:mc}) we have
\begin{align}
\mbb{\sup_{t\in[0,s]}\abs{n^{-1}H^n_t-t}\geq\eps}&\leq\eps^{-2}\meb{(n^{-1}H^n_s-s)^2} \nonumber \\
\nonumber \\
&=\frac{s}{\eps^2n}. \label{eq:Poisson-fishing}
\end{align}
To see that the latter inequality holds, observe  $\meb{n^{-1}H^n_s}=s$, so $\meb{(n^{-1}H^n_s-s)^2}$ equals the variance of $n^{-1}H^n_s$. Since $H_s^n$ has a Poisson distribution with rate $ns$,~\eqref{eq:Poisson-fishing} follows.

Note that the second (by~\eqref{eq:Poisson-fishing}) and third (by Theorem~\ref{thm:flln}) terms of the right-hand side of~\eqref{eq:split-soup} vanish in probability as $n\to\infty$. So we must have that $n^{-1}T^{n,d}_{\floor{tn}}\to x_d(t)$ uniformly in probability as $n\to\infty$. The convergence in $r$-mean follows since $n^{-1}T^{n,d}_{\floor{tn}}$ is bounded, see Lemma~\ref{lem:r-mean-conv} in Appendix~\ref{app:background-ODE-MT}.
\end{proof}

We connect the solution to preferences of students via the following lemma.

\lemProbToXPrime*
\begin{proof}

Fix $d \in \mathbb{N}$, $t,\eps>0$. We will show that there exists $N\inN$ such that for $n>N$,
\begin{align*}
\abs{\mbb{M^{n,d}_{\floor{tn}}=1}-x_d'(t)}<\eps.
\end{align*}
To proceed, observe that Lemma~\ref{lem:match-prob-approx} implies that there exists $N_1$ be such that for $n>N_1$ and all $x\in[0,1]$,
\begin{align*}
\abs{\mbb{M^{n,d}_{\floor{tn}}=1\given n^{-1}T^{n,d}_{\floor{tn}}=x}-(1-x^d)}& \leq \frac{\eps}{4}.
\end{align*}
Next, note that $1-x^d$ is continuous in $x$, so there exists some $\dta>0$ such that for all $y$ with $\abs{y-x}\leq\dta$, $\abs{(1-x^d)-(1-y^d)}<\eps/4$. Combining these two facts, we have that
\begin{align*}
\abs{\mbb{M^{n,d}_{\floor{tn}}=1\given n^{-1}T^{n,d}_{\floor{tn}}=y,\abs{y-x}\leq\dta}-(1-x^d)} & \leq \frac{\eps}{4} + \frac{\eps}{4} = \frac{\eps}{2}.
\end{align*}
By Theorem~\ref{thm:xts}, there exists now $N_2\inN$ such that for $n>N_2$,
\begin{align*}
\mbb{\abs{n^{-1}T^{n,d}_{\floor{tn}}-x_d(t)}>\dta}&<\frac{\eps}{2}.
\end{align*}
Let $N=\max\set{N_1,N_2}$, then putting all these together with the law of total probability, we have for $n>N$
\begin{align*}
\abs{\mbb{M^{n,d}_{\floor{tn}}=1}-x_d'(t)}&=\biggl|
\mbb{M^{n,d}_{\floor{tn}}=1\given\abs{n^{-1}T^{n,d}_{\floor{tn}}-x_d(t)}\leq\dta}\mbb{\abs{n^{-1}T^{n,d}_{\floor{tn}}-x_d(t)}\leq\dta} \\
&\qquad+\mbb{M^{n,d}_{\floor{tn}}=1\given\abs{n^{-1}T^{n,d}_{\floor{tn}}-x_d(t)}>\dta}\mbb{\abs{n^{-1}T^{n,d}_{\floor{tn}}-x_d(t)}>\dta}-x_d'(t)
\biggr| \\
&\leq\abs{\mbb{M^{n,d}_{\floor{tn}}=1\given n^{-1}T^{n,d}_{\floor{tn}}=y,\abs{y-x_d(t)}\leq\dta}+\mbb{\abs{n^{-1}T^{n,d}_{\floor{tn}}-x_d(t)}>\dta}-x_d'(t)} \\
&<\abs{1-x_d(t)^d+\frac{\eps}{2}+\frac{\eps}{2}-x_d'(t)} \\
&=\eps.
\end{align*}
\end{proof}

\subsection{Continuous Market}\label{sec:pf-main-result}
In this section we prove that the match rate in the continuous market increases with $d$ for all students $t\in[0,1]$. 

\thmBoundCts*
\begin{proof}
We use the notation $x(d,t)$ and simplify to $x$ when it is clear from the context. We wish to show that for all $1\leq d<\ell$
\begin{align*}
\frac{\partial}{\partial t}x(d,t)\leq \frac{\partial}{\partial t}x(\ell,t),
\end{align*}
which follows if for all $d\in[1,\infty)$ and $t\in[0,1]$ we have
\begin{align*}
\frac{\partial^2}{\partial d\partial t}x(d,t)\geq 0.
\end{align*}
Recall that the ODE is defined by $x(d,0)=0$, $\frac{\partial}{\partial t}x(d,t)=1-x^d$. Note that $x$ is twice continuously differentiable on its domain so its partial derivatives commute. We claim that $x(d,t)$ is a solution to this differential equation if and only if it solves the integral equation
\begin{align}\label{eq:timpl}
t&=\int_0^{x(d,t)}\frac{1}{1-u^d}\,du.
\end{align}
To observe this, note that $x(d,0)=0$, and by implicitly differentiating~\eqref{eq:timpl} with respect to $t$, we recover the required condition on the derivative of $x$ with respect to $t$:
\begin{align*}
\frac{\partial x}{\partial t}=1-x(d,t)^d.
\end{align*}
Now we can use the Leibniz integral rule to implicitly differentiate \eqref{eq:timpl} with respect to $d$ and obtain
\begin{align*}
0&=\frac{\partial}{\partial d}\br{\int_0^{x(d,t)}\frac{1}{1-u^d}\,du} \\
&=\int_0^{x(d,t)}\frac{\partial}{\partial d}\br{\frac{1}{1-u^d}}\,du+\frac{1}{1-x(d,t)^d}\frac{\partial x}{\partial d}.
\end{align*}
Rearranging for $\frac{\partial x}{\partial d}$ and computing the derivative in the integrand yields

\begin{align}
\frac{\partial x}{\partial d}&=-(1-x(d,t)^d)\int_0^{x(d,t)}\frac{u^d\log u}{(1-u^d)^2}\,du.\label{eq:dxdd} \\
&=-\frac{\partial x}{\partial t}\int_0^{x(d,t)}\frac{u^d\log u}{(1-u^d)^2}\,du.\nonumber
\end{align}

Note that $\frac{\partial x}{\partial t}$ is positive, and the integrand is negative (since $0\leq u\leq x\leq 1$, so $\log u\leq0$). This means $\frac{\partial x}{\partial d}\geq0$ (intuitively, fixing $t$ and increasing $d$ will increase the number of schools taken). Next compute
\begin{align*}
\frac{\partial^2x}{\partial d\partial t}&=\frac{\partial}{\partial d}(1-x(d,t)^d) \\
&=-x^d\log x-dx^{d-1}\frac{\partial x}{\partial d} \\
&=-x^{d-1}\br{x\log x+d\frac{\partial x}{\partial d}}.
\end{align*}

We want $\frac{\partial^2x}{\partial d\partial t}\geq 0$, which means we need $x\log x+d\frac{\partial x}{\partial d}\leq 0$. That is,

\begin{align*}
x\log x-d(1-x^d)\int_0^x\frac{u^d\log u}{(1-u^d)^2}\,du&\leq 0 \\
\iff \frac{x\log x}{1-x^d}-\int_0^x\frac{du^d\log u}{(1-u^d)^2}\,du&\leq 0.
\end{align*}

Expressing both terms as integrals, we get

\begin{align*}
\frac{x\log x}{1-x^d}-\int_0^x\frac{du^d\log u}{(1-u^d)^2}\,du&=\int_0^x\br{\frac{1}{1-u^d}+\frac{du^d\log u}{(1-u^d)^2}+\frac{\log u}{1-u^d}}\,du-\int_0^x\frac{du^d\log u}{(1-u^d)^2}\,du \\
&=\int_0^x\frac{1+\log u}{1-u^d}\,du.
\end{align*}

The rest of the proof is completed in Theorem~\ref{thm:ig} of Appendix~\ref{appx:lemmas_and_proofs} where we show that this integral is non-positive for $x(d,t)$ with $t\leq 1$.
\end{proof}

\subsection{Discrete Market: Probability of Being Matched}\label{sec:proof:discrete-market-probability}

In this section we prove the main results in the discrete market, often bringing back results from the continuous market via an application of Lemma~\ref{lem:prob-to-xprime}.

The following theorem succinctly states that all students in a balanced market (where $i\leq n$) prefer longer lists for $n$ large.

\thmMainDiscrete*
\begin{proof}
This follows directly from Theorem~\ref{thm:bound-cts} and Lemma~\ref{lem:prob-to-xprime}.
\end{proof}

\thmLemmaSchoolsLike*
\begin{proof}
We have 
\begin{align*}
\mbb{H_j^{n,d}(i)=1} & 
= \frac{1}{n} \sum_{r=1}^n \mbb{H^{n,d}_r(i)=1} & \hbox{(by symmetry)}
\\ & = \frac{1}{n} \sum_{r=1}^n \mathbb{E}(H^{n,d}_r({i})) \\ 
& = \frac{1}{n} \meb{\sum_{r=1}^n H^{n,d}_r(i)} \\ 
&=\frac{1}{n}\meb{T^{n,d}_i}. 
\end{align*}

Indeed, the random variable $\sum_{r=1}^n H^{n,d}_r(i)$ counts the number of schools matched to a student in positions $\{1,2,\dots,i-1\}$ in the the discrete market with $n$ schools, where every student has a preference list of length $d$. This equals the number of students from $\{1,2,\dots, i-1\}$ that are matched in that market, which is precisely $T^{n,d}_i$.

Recall from~\eqref{eq:discrete-match-prob} that
\begin{align*}
\mbb{M^{n,d}_i=1\given T^{n,d}_i=k}=\piecewise{1-\frac{\binom{k}{d}}{\binom{n}{d}}, & k\geq d, \\ 1, & \text{otherwise}.}
\end{align*}
It is clear from this expression that the left hand side monotonically increases in $d$: for all $i=1,2,3,\dots$, a longer list length increases the probability that student $i$ is matched to any school. Since $T_{i}^{n,d}=\sum_{\bar \imath =1}^{i-1} M_{\bar \imath}^{n,d}$ we deduce that $T^{n,\ell}_i\geq T^{n,d}_i$ in the sense of stochastic dominance, and therefore in expectation. 
\end{proof}

We next state a lemma that we prove in Appendix~\ref{appx:lemmas_and_proofs}.

\begin{restatable}{lemma}{lemXdBounds}\label{lem:xd-bounds}
For $d \geq 1$, we have $\br{\frac{2d+1}{4d+1}}^{1/d}\leq x(d,1)\leq\br{\frac{d+1}{2d+1}}^{1/d}$.
\end{restatable}

\thmCrossingDiscrete*
\begin{proof}
Consider $y(t)=x_\ell(t)-x_d(t)$, it is clear that $\lim_{t\to\infty}y(t)=0$. By Lemma~\ref{lem:xd-bounds}, it is also clear that $y(1)>0$. It therefore follows that there exists some $t\geq 1$ such that $y'(t)<0$. The thesis then follows from Lemma~\ref{lem:prob-to-xprime}. 
\end{proof}

We finally prove a corollary of Lemma~\ref{lem:xd-bounds} that is an intriguing and surprising result.
 
\thmBoundDiscrete*
\begin{proof}
From Lemma~\ref{lem:prob-to-xprime} and~\eqref{eq:ivp}, we have $\mbb{M_n^{n,d}=1}\to x'_d(1)= 1-x_d(1)^d$. The claim then follows from the bounds in Lemma~\ref{lem:xd-bounds}.
\end{proof}

\subsection{Discrete Market: Impact on Rank}\label{sec:proofs-impact-on-rank}

In this section, we prove the following lemma.

\lemWorstCaseRank*
\begin{proof}
Fix $d\inN$ and $k\leq d$. Similar to Lemma~\ref{lem:match-prob-approx}, it is easy to show that probability of getting matched to the first $k$ schools is given by
\begin{align*}
\mbb{K^{n,d}_i\leq k\given T^{n,d}_i=k}&=1-(k/n)^k\pm O(d^2/n).
\end{align*}
By Theorem~\ref{thm:xts}, we have that $n^{-1}T^{n,d}_i\stackrel{\mcP}{\to}x_d(i/n)$, so
\begin{align*}
\mbb{K^{n,d}_i\leq k}-\mbb{K^{n,d+1}_i\leq k}&\stackrel{\mcP}{\to}x_{d+1}(i/n)^k-x_d(i/n)^k.
\end{align*}
Consider the function $x_{d+1}(t)^k-x_d(t)^k$ for $t\in[0,1]$, and observe that
\begin{align*}
\frac{\partial}{\partial t}\br{x_{d+1}(t)^k-x_d(t)^k}&=k(x_{d+1}(t)^{k-1}x_{d+1}'(t)-x_d(t)^{k-1}x_d'(t)) \\
&\geq k(x_d(t)^{k-1}x_d'(t)-x_d(t)^{k-1}x_d'(t)) \\
&=0.
\end{align*}
The second last line follows from Theorem~\ref{thm:bound-cts} since $x_{d+1}'(t)\geq x_d'(t)$ and $x_{d+1}(t)\geq x_d(t)$ (from the former since $x_{d+1}(0)-x_d(0)=0$).

We therefore have that $x_{d+1}(i/n)^k-x_d(i/n)^k$ is increasing in $i$, so applying Lemma~\ref{lem:xd-bounds}, we have
\begin{align*}
\max_{i\leq n}\br{x_{d+1}(i/n)^k-x_d(i/n)^k}&=x_{d+1}(1)^k-x_d(1)^k \\
&\leq\br{\frac{d+2}{2d+3}}^{k/(d+1)}-\br{\frac{2d+1}{4d+1}}^{k/d},
\end{align*}
which completes the proof.
\end{proof}

\section{Numerical Experiments}\label{sec:numerical}

We perform numerical experiments to evaluate how well our results generalize beyond the one-to-one Serial Dictatorship mechanism with uniformly random preferences.

\paragraph{Schools not being sampled uniformly at random.} In our first set of numerical experiments we investigate the impact of students not picking their preference list uniformly at random from all schools. In particular, we now assume there is some probability distribution $p:\set{1,\dots,n}\to[0,1]$ with $\sum_{j=1}^n p(j)=1$, with $p(j)$ dictating the probability that students sample school $j$ at their turn. The case of uniform sampling occurs when $p\equiv 1/n$.

We experiment with five different distributions for $p$: a uniform distribution ($p_1$), two types of Pareto-like distributions where students' preference of schools is concentrated at certain schools ($p_2$ and $p_3$ with low, and high concentration respectively), a setting of two classes (high demand and low demand) of schools ($p_4$), and finally a ``degenerate'' distribution where students sample from the first half of schools almost exclusively ($p_5$). Note that we would expect the last distribution to behave as if we had half as many schools, so that only students with $i\leq n/2$ would be guaranteed to prefer longer lists. The distributions are as follows:

\begin{align*}
p_1(j)&=\frac{1}{n}, \qquad p_2(j)\propto\frac{1}{2+j/n}, \qquad p_3(j)\propto\frac{1}{(1+j/n)^{10}}, \\
& \\
p_4(j)&\propto1+3\indc_{\set{j\leq 200}}, \qquad p_5(j)=\piecewise{2/n-1/100,&j\leq 500,\\ 1/100,&j>500.}
\end{align*}

We perform all numerics with $n=1000$ schools, and compare the cases of $d$ being $1,2,4,10$, and $20$. All experiments are done with 100 000 repetitions.

Figure~\ref{fig:nonuniform-numerics} in Appendix~\ref{appx:numerical-experiments} shows the outcome of these experiments. We graph first the distribution of $p(\cdot)$ for all schools on the left, then the proportion of schools taken by students up to student $i$ in the middle, and the probability that a given student $i$ is matched to a school on the right.

Our main thesis -- that all students in balanced markets prefer long lists -- holds under all distributions other than the degenerate one (which is almost equivalent to having half as many schools). In particular, $\mbb{M^{n,d}_i=1}$ increases in $d$ under all distributions other than the degenerate. We conclude that our main result, Theorem~\ref{thm:main-discrete}, seems to be robust to i.i.d.~sampling via reasonable non-uniform distributions.

\paragraph{Schools with multiple seats.} In Section~\ref{para:multiple-seats} and Appendix~\ref{appx:q_ge_1} we discuss the extension of the discrete and continuous models to the case when schools have $q\inN$ seats. We now perform some numerical experiments to verify Conjecture~\ref{conj:q_ge_1}. To do so, we compute numerical solutions to the differential equations describing the continuous model using the Euler method with step size $5\times10^{-5}$, and compare the probability of getting matched for successive values of $d\inN$. We verified that the conjecture holds for all pairs of $q=1,2,\dots,20$ and $d=1,2,\dots,15$. This leads us to believe that the conjecture holds for all $q\inN$ and $d\geq 1$. In Appendix~\ref{appx:numerical-experiments} we show some plots of the resultant output.

\section{Conclusions}\label{sec:conclusions}

In this paper, we investigated the impact of truncating preference lists in a two-sided matching market where students choose schools following a  Serial Dictatorship order. Our main result is that if the market is balanced, all agents increase their probability of being matched when lists are longer. This result is proved for preference lists that are sampled uniformly at random and school quotas equal to $1$ and then shown to hold numerically for more complex preferences and larger quotas. We believe these valuable insights can be used to support the expansion of the length of preference lists, which are often set to small values. Investigating similar questions for other mechanisms is a relevant open question.

\paragraph{Acknowledgments.} The authors would like to acknowledge the support of NSF Grant 2046146 and the valuable feedback provided by three anonymous EC'25 reviewers.

\bibliographystyle{ACM-Reference-Format}
\bibliography{refs}

\appendix

\section{Background on Differential Equations and Stochastic Processes}\label{app:background-ODE-MT}

This appendix contains some preliminaries that are required to follow the proofs in the main text, in particular on the theory of ordinary differential equations (ODEs) and initival value problems (IVPs), as well as probability and Markov theory.

\subsection*{Ordinary Differential Equations}\label{appx:ode}

An ordinary differential equation is a description of the local dynamics of a system in terms of its current state and derivatives; in particular, a first order ordinary differential equation describes the rate of change of a quantity $x$ as a function of the current time $t$ and the current state $x(t)$. That is, $x'(t)=f(t,x(t))$ for some $f$. A problem with such a description along with boundary data (e.g., $x(t_0)=x_0$) is called an \emph{initial value problem}. A main theme in differential equations is to understand when such local descriptions give rise to solutions, either in some small neighborhood of a point, or globally. Lipschitz continuity is an important property in differential equations, as it is a sufficient condition for the existence of such solutions as we will argue next.

\begin{definition}[Lipschitz continuity]
A function $f:\bbR^m\to\bbR^n$ is \emph{Lipschitz continuous} under norm $\norm{\cdot}$ on $U\subseteq\bbR^m$ if there exists some $L\geq 0$ such that
\begin{align*}
\norm{f(x)-f(y)}\leq L\norm{x-y},
\end{align*}
for all $x,y\in U$. We call $L$ the Lipschitz constant of $f$ on $U$.
\end{definition}

The standard result for the existence and uniqueness of solutions to initial value problems of ordinary differential equations comes from variations of the Picard-Lindel\"of Theorem, which roughly state that if $f$ is Lipschitz at a point, a solution exists and is unique in some neighborhood of that point. The following is a version that extends this fact to the positive reals for the one dimensional case, adapted from \cite[Theorem~2.2 and Corollary~2.6]{teschl2024ordinary}.

\begin{theorem}\label{thm:picardlindelof}
Consider the initial value problem
\begin{align}\label{eq:ivp-appendix}
    x'=f(t,x), \quad x(t_0)=x_0,
\end{align} where $(t_0,x_0)\in U$ for some open set $U \subseteq \bbR^2$ and $f:U\to\bbR$ is a Lipschitz continuous function. Then if $[t_0,\infty)\cross\bbR\in U$ there exists a unique solution $x(t)$ to~\eqref{eq:ivp-appendix} for all $t\geq t_0$.
\end{theorem}

\subsection*{Probability and Markov theory}\label{appx:mc}

In this section we briefly cover some fundamentals of general probability and Markov theory. The treatment is not fully formal as we omit certain regularity conditions and assumptions -- however such conditions are immediately satisfied by the objects we study in this paper. We refer to \cite{grimmett2020probability} for further details.

Recall that a random variable $X$ on a probability triple $(\Omega,\mcF,\mb)$ is a measurable real-valued function from $(\Omega,\mcF)$ to $(\bbR,\mcL)$.

\begin{definition}[Stochastic process]
A stochastic process on a state space $\mcS\subseteq\bbR$ is a set of random variables $\set{X_t}_{t\in\mcT}$ on the same probability triple that are indexed by some $\mcT$ (thought of as time) and each take values in $\mcS$ (that is, $X_t\in\mcS$ for all $t\in\mcT$).
\end{definition}

Common cases are $\mcT=\set{0,1,2,3,\dots}$ or $\mcT=\set{0,\Delta,2\Delta,3\Delta,\dots}$ for some $\Delta>0$ for \emph{discrete time} processes and $\mcT=[0,\infty)$ for \emph{continuous time} processes.

\paragraph{Convergence.} Consider a sequence $\set{X^n}_{n=1,2,\dots}$ of random variables on the same probability triple $(\Omega,\mcF,\mb)$. We recall some ways in which $X^n$ may converge to some other random variable $X$.

\begin{definition}[Convergence in Probability]
The sequence of random variables $\{X^n\}_{n=1,2,\dots}$ converges to the random variable $X$ in probability, denoted $X^n\stackrel{\mcP}{\to}X$ if for all $\eps>0$,
\begin{align*}
\lim_{n\to\infty}\mbb{\abs{X^n-X}\geq\eps}\to 0,
\end{align*}
as $n\to\infty$.
\end{definition}

A stochastic process can converge pointwise (for all $t\in\mcT$) in probability, but clearly a stronger condition is that of uniform convergence, where this convergence is uniform through $\mcT$.

\begin{definition}[Uniform convergence in probability]
The sequence of stochastic processes $\set{X^n_t}_{t\in\mcT,n=1,2,\dots}$ on $\mcT$ converges uniformly in probability to $\set{X_t}_{t\in\mcT}$ if for all $\eps>0$,
\begin{align*}
\mbb{\sup_{t\in\mcT}\abs{X^n_t-X_t}\geq\eps}\to0,
\end{align*}
as $n\to\infty$.
\end{definition}

That is, the maximum deviation over the index set between the $X^n_t$ and $X_t$ is itself bounded and vanishes in probability.

\begin{definition}[Convergence in $r$-mean]
The sequence of random variables $\{X^n\}_{n=1,2,\dots}$ converges to a random variable $X$ in $r$-mean, denoted $X^n\stackrel{L^r}{\to}X$ if for all $\eps>0$, $\lim_{n\to\infty}\meb{\abs{X^n-X}^r}\to 0$, as $n\to\infty$.
\end{definition}

Note in particular that for $r=1$, this is convergence in mean, that is, $\meb{X^n}\to\meb{X}$. Convergence in probability does not automatically imply convergence in $r$-mean, though the converse is true. One special case where the converse holds is when $X^n$ and $X$ are bounded.

\begin{lemma}\label{lem:r-mean-conv}
Suppose $X^n\stackrel{\mcP}{\to}X$ and $\abs{X}\leq M$ and $\abs{X^n}\leq M$ for some $M\geq 0$. Then $X^n\stackrel{L^r}{\to}X$ for all $r\geq 1$.
\end{lemma}

\paragraph{Markov processes.} A Markov process is a stochastic process which additionally satisfies a property of being \emph{memoryless}, meaning that its future evolution is entirely dictated by its current state independent of the past. We restrict ourselves to time-homogeneous Markov chains with finite state spaces.

\begin{definition}[Discrete time Markov chain]
A \emph{discrete time Markov chain} is a stochastic process on a finite state space $\mcS\subseteq\bbR$ and index set $\mcT=\set{0,1,2,3,\dots}$ that satisfies the property
\begin{align*}
\mbb{X_{t+1}=j\given X_t=i_t,X_{t-1}=i_{t-1},\dots,X_1=i_1}=\mbb{X_{t+1}=j\given X_t=i_t},
\end{align*}
for all states $j,i_\cdot\in\mcS$ and all $t\in \mcT$. The evolution of such a Markov chain is therefore determined by its (one step) \emph{transition probabilities} $p(i,j)$ for $i,j\in\mcS$, arranged in a transition matrix $P$ defined by
\begin{align*}
p(i,j)=\mbb{X_{t+1}=j\given X_t=i}.
\end{align*}
\end{definition}

\begin{definition}[Continuous time Markov chain]
A \emph{continuous time Markov chain} is a stochastic process on a finite state space $\mcS\subseteq\bbR$ and index set $\mcT=[0,\infty)$ that satisfies the property
\begin{align}\label{eq:continuous-MC}
\mbb{X_{t+s}=j\given X_t=i,X_{t_l}=i_l,X_{t_{l-1}}=i_{l-1},\dots,X_0=i_0}&=\mbb{X_s=j\given X_0=i},
\end{align}
for all $s\geq 0$, $t>t_l>t_{l-1}>\cdots>t_1>t_0\geq 0$ and all $j,i,i_\cdot\in\mcS$.
\end{definition}

Note that~\eqref{eq:continuous-MC} implies that the transition times must also be random and memoryless. It turns out that this property embodies continuous time Markov chains with a large amount of structure. In particular, the time between transitions (when the chain moves from one state to another) must be exponentially distributed (this is the only memoryless continuous distribution, satisfying $\mbb{X\leq x\given X\geq y}=\mbb{X\leq x-y}$).

In the discrete case the evolution was defined by the one step transition probabilities $P$. In the continuous case, these are replaced by a Matrix-valued function of time $P(t)$. It turns out there is a compact way to represent the possible transitions via \emph{transition rates} $q(i,j)$ that describe the instantaneous rate of moving from state $i$ to state $j$ as formalized by the following lemma.

\begin{theorem}[Transition rates]
For a continuous time Markov chain on a finite state space, let $p(i,j;t)=\mbb{X_t=j\given X_0=i}$ for $t\geq 0$. Denote by $P(t)$ the matrix with entries $p(i,j;t)$. Then there exists a \emph{transition rate} matrix $Q$ with entries $q(i,j)$ that is the unique solution to
\begin{align*}
Q=\lim_{t\to0}\frac{P(t)-I}{t},
\end{align*}
note that $q(i,i)=-\sum_{j\neq i}q(i,j)$. 

\end{theorem}

We next introduce the concept of a Poisson point process in order to build continuous time Markov chains from discrete time ones. Intuitively, a homogeneous Poisson point process simply counts the number of events that have happened until time $t$, where the time between successive events is iid and exponentially distributed.

\begin{definition}[Homogeneous Poisson point process]
Let $\lambda>0$. Define the stochastic process $\{H_t\}_{t\geq 0}$ with state space $\set{0,1,2,\dots}$ as follows. For $i=1,2,\dots$ let $E_i\distributed\Exp(\lambda)$, and let $H_t=\sup_{i=0,1,2,\dots}\set{\sum_{j=1}^iE_j\leq t}$. We call $\{H_t\}_{t\geq 0}$ the homogeneous Poisson process with rate $\lambda$ on $t\geq 0$. Note that $H_t\distributed\Pois(\lambda t)$, and in particular $\meb{H_t}=\lambda t$ for all $t\geq 0$.
\end{definition}

A Homogeneous Poisson point process is an example of a very simple Markov chain.

\begin{lemma}
A homogeneous Poisson point process with rate $\lambda$ is a continuous time Markov chain on the countable state space $\mcS=\set{0,1,2,\dots}$ with transition rates
\begin{align*}
q(i,j)&=\piecewise{
\lambda,& j=i+1, \\
-\lambda,& j=i, \\
0,& \text{otherwise}.
}
\end{align*}
\end{lemma}

Given a continuous time Markov chain, one can decompose it into a discrete time Markov chain that describes the transition probabilities at jump times, and a per-state rate of an Exponential distribution that describes the duration that the chain stays in that state before jumping to the next.

The following theorem gives one variation of a theorem that allows easily constructing continuous time Markov chains from discrete time Markov chains.

\begin{theorem}[Constant-rate Markov Chain Embedding]\label{thm:embedding}
Let $\set{X_n}_{n=0,1,2,\dots}$ be a discrete time Markov chain on finite state space $\mcS$ with $X_0=s_0$ for some initial state $s_0\in\mcS$. Let $\lambda>0$ be some rate and let $\{H_t\}_{t\geq 0}$ be the homogeneous Poisson point process with rate $\lambda$ on $[0,\infty)$.

Define a stochastic process $\set{Z_t}_{t\geq 0}$ 
by $Z_t=X_{H_t}$. Then $\set{Z_t}_{t\geq0}$ is a continuous time Markov chain whose state space is $\mcS$, initial state is $s_0$, and that has rates $q(i,j)=\lambda p(i,j)$ for $i\neq j$ and $q(i,i)=-\lambda(1-p(i,i))$. This continuous time Markov chain is called an \emph{embedded} chain because its transition rates follow the transition probabilities of the discrete chain, that is, $\mbb{Z_s=j\given Z_t=i}=p(i,j)$ if $H_t=H_s+1$.
\end{theorem}

\paragraph{Martingales.} A martingale is a stochastic process with the property that its expected value in the future is the current value.

\begin{definition}[Martingale]
A stochastic process $\{X_t\}_{t\geq 0}$ is a \emph{martingale} if \, $\meb{X_s\given X_t}=X_t$ for all $s\geq t$.
\end{definition}

The homogeneous Poisson process corrected for its mean is a martingale. That is, $H_t-\lambda t$ is a martingale. The following is a standard result in martingale theory.

\begin{theorem}[Doob's martingale inequality]\label{thm:doob}
If $\{X_t\}_{t\geq 0}$ is a martingale, then for all $T\geq 0$ and $C>0$,
\begin{align*}
\mbb{\sup_{t\in[0,T]}X_t\geq C}\leq \frac{\meb{\max(X_T,0)}}{C}.
\end{align*}
Further, for $r\geq 1$, we have
\begin{align*}
\mbb{\sup_{t\in[0,T]}\abs{X_t}\geq C}\leq\frac{\meb{\abs{X_T}^r}}{C^r}.
\end{align*}
\end{theorem}

\section{Extension to Schools Having Multiple Seats}\label{appx:q_ge_1}

In Section~\ref{para:multiple-seats}, we set up an analogue of the discrete model of Section~\ref{sec:discrete} in the case that each school has $q=1,2,3,\dots$ seats. In this appendix, we roughly trace the steps of Section~\ref{sec:models} for the case of arbitrary $q$. We first describe the continuous model for the multiple-seat case, then discuss how Theorem~\ref{thm:xts} may be extended to this case. In particular, we develop the initial value problem equivalents of~\eqref{eq:ivp} for arbitrary $q=1,2,3,\dots$ and $d\geq 1$. We discuss the solutions to the initial value problem for $d=1$, then discuss how we may use this solution for $d>1$.

\paragraph{The continuous model with multiple seats.} Recall that $T^{n,d}_i$ is the number of schools with $q$ seats taken after students $\set{1,2,3,\dots,i-1}$ have had their turn, and $S^k_i$ is the number of schools with $k$ seats taken at this point for $k=1,2,\dots,q-1$. We again keep track of the continuous analogues, for $i=\floor{tn}$, letting $x_d(t)$ be the proportion of schools with all seats taken, and $y_d^k(t)$ be the proportion of schools with $k$ seats taken for $k=0,\dots,q-1$. Initially all schools have no seats taken, so
\begin{align}
x_d(0)=0, \qquad y_d^0(0)=1, \qquad y_d^k(0)=0,\quad k=1,2,\dots,q-1. \label{eq:q_ge_1_iv}
\end{align}
As we discuss in the main body, the probability that student $t$ is matched to any student goes to $1-x_d(t)^d$ in probability, and the proportion of time they get matched to a school with $k$ seats taken is proportional to $S^k_i$, we therefore get that the rate at which schools with $k$ seats taken become schools with $k+1$ seats taken is
\begin{align*}
(1-x_d(t)^d)\frac{y_d^k(t)}{\sum_{m=0}^{q-1}y_d^m(t)}&=\frac{1-x_d(t)^d}{1-x_d(t)}y_d^k(t),
\end{align*}
since $x_d(t)+\sum_{k=0}^{q-1}y_d^k(t)=1$. Note further that this is a positive flow into $y_d^{k+1}$ and negative for $y_d^k$. We have each of these flows from $y_d^k$ to $y_d^{k+1}$ for $k=0,\dots,q-1$ and one from $y_d^{q-1}$ to $x_d$. For convenience, define
\begin{align*}
\gamma_d(t)=\frac{1-x_d(t)^d}{1-x_d(t)}.
\end{align*}
Denoting derivative with respect to time with a dot for clarity to avoid multiple superscripts, this gives us the following differential equation
\begin{align}
\dot y_d^0(t)&=-\gamma_d(t)y_d^0(t), \nonumber \\
\dot y_d^k(t)&=\gamma_d(t)(y_d^{k-1}(t)-y_d^k(t)),\qquad k=1,\dots,q-1, \label{eq:q_ge_1_ode} \\
\dot x_d(t)&=\gamma_d(t)y_d^q(t). \nonumber
\end{align}

\paragraph{Connection with the discrete model.} Following the steps in Section~\ref{sec:continuous-to-discrete}, we can apply a slightly more general equivalent of Theorem~\ref{thm:xts} for this multi-dimensional case, again from \cite[Theorem~17.3.1]{kroese2013handbook} to prove the following lemma.

\begin{lemma}\label{lem:xts-q_ge_1}
Fix $d\in\bbN$, $q=1,2,\dots$ and define the initial value problem for $\{x_d,y_d^k\}$ for $k=0,1,\dots,q-1$ via~\eqref{eq:q_ge_1_iv} and~\eqref{eq:q_ge_1_ode}. For $t\geq 0$, let $T^{n,d}_i$ and $S^k_i$ for $k=0,\dots,q-1$ be as in Section~\ref{para:multiple-seats}. Then $n^{-1}T^{n,d}_{\floor{tn}}\to x_d(t)$ and $n^{-1}S^k_{\floor{tn}}\to y_d^k(t)$ for $k=0,1,\dots,q-1$ uniformly in probability as $n\to\infty$, where $\{x_d(t),y_d^k(t)\}$ for $k=0,1,\dots,q-1$ is the unique solution satisfying the initial value problem of \eqref{eq:q_ge_1_iv} and~\eqref{eq:q_ge_1_ode} for $t\geq 0$.
\end{lemma}
\begin{proof}
The proof follows identically to the proof of Theorem~\ref{thm:xts} in Section~\ref{sec:proofs}, adapted to the case of multiple equations.
\end{proof}

\paragraph{Dynamics of the multiple seats model for $d=1$.} Observe that for $d=1$, we have $\gamma_d\equiv 1$. As in the earlier continuous model, we can again solve this explicitly. It's easy to see we get $y_1^0(t)=\exp(-t)$, and for $k=1$, we now have $\dot y_1^1=\exp(-t)-y_1^1$, which yields $y_1^1(t)=t\exp(-t)$. Following this pattern, we have
\begin{align}
y_1^k(t)&=\frac{t^k}{k!}e^{-t}, \qquad k=0,\dots,q-1, \nonumber \\
x_1(t)&=1-e^{-t}\sum_{k=0}^{q-1}\frac{t^k}{k!}.\label{eq:q_ge_1_sol_d1}
\end{align}

Readers familiar with probability or phase-type distributions will note that this is the cumulative distribution function of an Erlang distribution with parameters $k=q$ and $\lambda=1$. Figure~\ref{fig:q_ge_1_sample} shows the solution for $d=1$, $q=4$.

\begin{figure}[!h]
    \centering
    \includegraphics[width=0.75\linewidth]{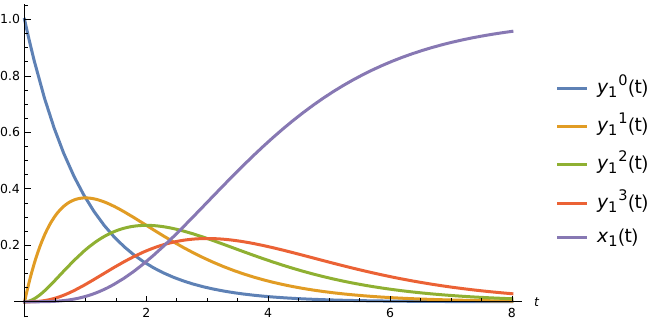}
    \caption{Solution to the multiple-seat continuous market for $d=1$, $q=4$.}
    \label{fig:q_ge_1_sample}
\end{figure}

Recall finally that a student $t$ gets matched with probability that goes to $1-x_d(t)^d$. Denote by $M^{n,d,q}_{i}$ be the indicator random variable of whether a student $i$ gets matched with $n$ schools, $d$ length lists, and $q$ seats per school. We therefore have
\begin{align*}
\mbb{M^{n,d,q}_i=1}\stackrel{\mcP}{\to}1-x_d(t)^d,
\end{align*}
as $n\to\infty$. This follows similarly to Lemma~\ref{lem:prob-to-xprime} via Lemma~\ref{lem:xts-q_ge_1}. Note that the difference to Lemma~\ref{lem:prob-to-xprime} is simply that this time $x_d'(t)$ is not necessarily equal to $1-x_d(t)^d$.

\paragraph{Discussion of dynamics for $q>1$.} The fact that $\gamma_d$ multiplies every equation in~\eqref{eq:q_ge_1_ode} suggests the following idea. Define an initial value problem for $\tau_d$ as
\begin{align}\label{eq:tau_ivp}
\tau_d(0)&=0, \qquad 
\dot\tau_d(t)=\frac{1-x_1(\tau_d(t))^d}{1-x_1(\tau_d(t))}.
\end{align}
It is not hard to show that the following now holds
\begin{align}
y_d^k(t)&=y_1^k(\tau_d(t)), \qquad k=0,1,\dots,q-1 \nonumber \\
x_d(t)&=x_1(\tau_d(t)). \label{eq:q_ge_1_sol_dgen}
\end{align}
This now suggests a procedure for computing solutions to the initial value problem defined by~\eqref{eq:q_ge_1_iv} and~\eqref{eq:q_ge_1_ode}: first compute the solution to~\eqref{eq:q_ge_1_sol_d1}, then solve~\eqref{eq:tau_ivp} to compute $\tau_d(t)$ and plug it into~\eqref{eq:q_ge_1_sol_dgen}.

We further remark that one can interpret $\tau_d(x)$ as describing the ratio of how much faster students are matched to schools when they get $d>1$ tries compared to a single try. $1-x_1^d$ is the probability of getting matched to any school given a list of length $d$ and $1-x_1$ is that probability with lists of length $1$. The interpretation via~\eqref{eq:q_ge_1_sol_dgen} then tells us that the case for general $d$ is exactly the same as for $d=1$ except now the rate at which students are matched to schools is rescaled by a factor of $\tau_d$.

\section{Missing Proofs }\label{appx:lemmas_and_proofs}

\paragraph{From Section~\ref{sec:implications-extensions}.}

\corSerial*

\begin{proof}
By symmetry, we have $$\mbb{R^{n,d}=1}= \frac{n}{m} \mbb{H^{n,d}_1(m)=1},$$ where $m$ is the number of students in the market. From Lemma~\ref{lem:school-love}, $\mbb{H^{n,d}_1(m)=1}$ increases with $d$, so $\mbb{R^{n,d}=1}$ also increases with $d$.
\end{proof}

\begin{theorem}\label{thm:ig}
Suppose $x(d,t)\geq 0$ is defined implicitly via
\begin{align*}
\int_0^x\frac{1}{1-v^d}\,dv=t.
\end{align*}
Then for $t \in [0,1]$, $d\geq 1$, we have
\begin{align}
\int_0^x\frac{1+\log u}{1-u^d}\,du\leq 0.\label{eq:ig-to-prover}
\end{align}
\end{theorem}
\begin{proof}
Observe that $x(d,t)$ is monotonically increasing in $t$ and that $x(d,t)\leq t$. The integrand in \eqref{eq:ig-to-prover} is negative on $u\in[0,1/e]$ and vanishes at $u=1/e$ after which it increases monotonically. The results therefore follows immediately when $x(d,t)\in[0,1/e]$, and since the integral increases for $x\in[1/e,1)$, it suffices to show the result for $x(d,1)$.
In Lemma~\ref{lem:xd-bounds} we show that
\begin{align*}
x(d,1)\leq\br{\frac{d+1}{2d+1}}^{1/d},
\end{align*}
for $d\geq 1$ and in Lemma~\ref{lemma:x1a_ineq} we show
\begin{align*}
\int_0^{\br{\frac{d+1}{2d+1}}^{1/d}}\frac{1+\log(u)}{1-u^d}\,du&\leq0,
\end{align*}
for $d \geq 1$ which completes the proof.
\end{proof}

\lemXdBounds*
\begin{proof}
We first show the upper bound. Recall $x(d,1)$ is defined via
\begin{align*}
1=\int_0^{x(d,1)}\frac{du}{1-u^d}.
\end{align*}
Since $(1-u^d)$ is strictly positive, it suffices to show
\begin{align*}
\int_0^{x(d,1)}\frac{du}{1-u^d}\leq \int_0^{\br{\frac{d+1}{2d+1}}^{1/d}}\frac{du}{1-u^d},
\end{align*}
or equivalently
\begin{align*}
1\leq \int_0^{\br{\frac{d+1}{2d+1}}^{1/d}}\frac{du}{1-u^d}.
\end{align*}

$(1-u^d)^{-1}$ is the sum of a geometric series for $u^d\in(0,1)$, so with $z\in(0,1)$ write

\begin{align}
\int_0^z\frac{du}{1-u^d}&=\int_0^z\sum_{r=0}^\infty u^{rd}\,du \nonumber \\
&=\sum_{r=0}^\infty\int_0^zu^{rd}\,du \nonumber \\
&=\sum_{r=0}^\infty\frac{z^{rd+1}}{rd+1} \nonumber \\
&=z\br{1+\sum_{r=1}^\infty\frac{(z^d)^r}{rd+1}}\label{eq:x1a-mid} \\
&\geq z\br{1+\sum_{r=1}^\infty\frac{(z^d)^r}{r(d+1)}} \nonumber \\
&=z\br{1-\frac{1}{d+1}\log(1-z^d)}.\label{eq:x1a-last}
\end{align}

The validity of exchanging the order of summation and integration follows by the Tonelli theorem since the integrand is non-negative. Now substitute $z=\br{\frac{d+1}{2d+1}}^{1/d}$ and $q=1/d$. We have $q\in(0,1]$ and $z=\br{\frac{q+1}{q+2}}^q$, and we need to show that \eqref{eq:x1a-last} is lower bounded by $1$. After some rearranging, this reduces to showing
\begin{align*}
\log(2+q)&\geq\br{\frac{q+1}{q}}\br{\br{1+\frac{1}{q+1}}^q-1},
\end{align*}
for $q \in (0,1]$, which we show in Lemma~\ref{lemma:first_ineq}.

For the lower bound, we apply identical reasoning until~\eqref{eq:x1a-mid}, but must show that for $z=\br{\frac{2d+1}{4d+1}}^{1/d}$, we have
\begin{align*}
1\geq z\br{1+\sum_{r=1}^\infty\frac{(z^d)^r}{rd+1}}.
\end{align*}
We again substitute $q=1/d$ and this time get $z=\br{\frac{q+2}{q+4}}^q$, and this yields
\begin{align*}
z\br{1+\sum_{r=1}^\infty\frac{(z^d)^r}{rd+1}}&=z\br{1+\frac{z^d}{d+1}+\sum_{r=2}^\infty\frac{(z^d)^r}{rd+1}} \\
&\leq z\br{1+\frac{z^d}{d+1}+\sum_{r=2}^\infty\frac{(z^d)^r}{rd}} \\
&=z\br{1+\frac{z^d}{d+1}-\frac{z^d}{d}-\frac{1}{d}\log(1-z^d)} \\
&=z\br{1-\frac{z^d}{d(d+1)}-\frac{1}{d}\log(1-z^d)} \\
&=\br{\frac{q+2}{q+4}}^q\br{1-\frac{q^2(q+2)}{(q+1)(q+4)}-q\log\br{\frac{2}{q+4}}}.
\end{align*}
To complete the proof, we therefore need
\begin{align*}
\br{\frac{q+2}{q+4}}^q\br{1-\frac{q^2(q+2)}{(q+1)(q+4)}-q\log\br{\frac{2}{q+4}}}&\leq 1
\end{align*}
for $q \in (0,1]$, which we show in Lemma~\ref{lemma:technical_2}.
\end{proof}

\begin{lemma}\label{lemma:technical_2}
$1-\frac{q^2(q+2)}{(q+1)(q+4)}-q\log\br{\frac{2}{q+4}}\leq \br{\frac{q+4}{q+2}}^q$ for $q\in[0,1]$.
\end{lemma}
\begin{proof}

\begin{align*}
&1-\frac{q^2(q+2)}{(q+1)(q+4)}-q\log\br{\frac{2}{q+4}}\leq \br{\frac{q+4}{q+2}}^q \\
& \iff \br{\frac{q+4}{q+2}}^q+\frac{q^2(q+2)}{(q+1)(q+4)}-q\log\br{2+\frac{q}{2}}-1\geq 0 \\
&\iff \br{\frac{q+4}{q+2}}\br{\frac{q+4}{q+2}}^{q-1}+\frac{q^2(q+2)}{(q+1)(q+4)}-q\log\br{2+\frac{q}{2}}-1\geq 0.
\end{align*}

Now apply Claim~\ref{claim:ineqs2} and Claim~\ref{claim:two-ineqs} to get
\begin{align*}
&\br{\frac{q+4}{q+2}}\br{\frac{q+4}{q+2}}^{q-1}+\frac{q^2(q+2)}{(q+1)(q+4)}-q\log\br{2+\frac{q}{2}}-1 \\
&\qquad\geq\br{\frac{q+4}{q+2}}\br{\frac{1}{2}+\frac{q}{8}(1+4\log(2))}+\frac{q^2(q+2)}{(q+1)(q+4)}-q\br{\frac{1}{4} \left(1-\frac{q}{10}\right) q+\log (2)}-1 \\
&\qquad=\frac{q^2 \left(q^4-3 q^3-11 q^2-20 q^2 \log (2)+53 q-100 q \log (2)+100-80 \log (2)\right)}{40 (q+1) (q+2) (q+4)}
\end{align*}

This is positive if $q^4-3 q^3-11 q^2-20 q^2 \log (2)+53 q-100 q \log (2)+100-80 \log (2)\geq 0$, but this is a polynomial with two real roots, both at $q>1$, so since this equation is positive for $q=0$, it must be non-negative for all $q\in[0,1]$.
\end{proof}

\begin{claim}\label{claim:ineqs2}
$\left(\frac{q+4}{q+2}\right)^{q-1}\geq\frac{1}{2}+\frac{q}{8}(1+4\log(2))$ for $q\in[0,1]$.
\end{claim}
\begin{proof}
Observe this holds at equality for $q=0$ and apply Claim~\ref{claim:two-ineqs} to the the derivative of the difference to get:
\begin{align*}
\frac{\partial}{\partial q}\sbr{
\br{\frac{q+4}{q+2}}^{q-1}-\frac{1}{2}-\frac{q}{8}(1+4\log(2))}&\geq\frac{1}{16}\br{(q+2)(q+4)\log\br{\frac{q+4}{q+2}}-2(q+4\log(2))}.
\end{align*}
This is now true if and only if
\begin{align}\label{eq:name}
\log\br{\frac{q+4}{q+2}}-\frac{2(q+4\log(2))}{(q+2)(q+4)}\geq 0.
\end{align}
One easily checks
\begin{align*}
\frac{\partial}{\partial q}\sbr{\log\br{\frac{q+4}{q+2}}-\frac{2(q+4\log(2))}{(q+2)(q+4)}}=\frac{4 (-3 q+4 q \log (2)-8+12 \log (2))}{(q+2)^2 (q+4)^2}\geq 0,
\end{align*}
where the inequality holds since both the numerator and the denominator are nonnegative. Since one readily verifies that~\eqref{eq:name} is valid for $q=0$, the claim follows.
\end{proof}

\begin{claim}\label{claim:two-ineqs}
$\br{\frac{q+4}{q+2}}^q\geq\frac{(q+4)^2}{16}$ and $\log\left(2+\frac{q}{2}\right)\leq\log(2)+\frac{1}{4}\left(1-\frac{q}{10}\right)q$ for $q\in[0,1]$.
\end{claim}
\begin{proof}

For the first:
\begin{align*}
\br{\frac{q+4}{q+2}}^q&\geq\frac{(q+4)^2}{16} \\
\iff \br{\frac{q+4}{q+2}}^q&\geq\br{1+\frac{q}{4}}^2 \\
\iff \log\br{\frac{q+4}{q+2}}&\geq \frac{2}{q}\log\br{1+\frac{q}{4}}.
\end{align*}
Note that the left hand side always exceeds $\log(5/3)>1/2$, and the right hand side never exceeds $1/2$ because $\log(1+q/4)\leq q/4$.

For the second, take the derivative of the difference of the left term and the right term:
\begin{align*}
\frac{\partial}{\partial q}\br{\log\left(2+\frac{q}{2}\right)-\frac{1}{4}\left(1-\frac{q}{10}\right)q-\log (2)}&=\frac{1}{q+4}+\frac{q-5}{20} \\
&=\frac{(q-1)q}{20(q+4)},
\end{align*}
which is now clearly non-positive.
\end{proof}

\begin{lemma}\label{lemma:x1a_ineq}
For $d\geq 1$, we have
\begin{align*}
\int_0^{\br{\frac{d+1}{2d+1}}^{1/d}}\frac{1+\log(u)}{1-u^d}\,du&\leq0.
\end{align*}
\end{lemma}
\begin{proof}
Change variables with $q=1/d$ (so $q\in[0,1]$) and define $A=\frac{d+1}{2d+1}=\frac{q+1}{q+2}$, then write
\begin{align*}
\int_0^{A^q}\frac{1+\log(u)}{1-u^{1/q}}\,du&=\int_0^{A^q}\sum_{n=0}^\infty u^{n/q}(1+\log(u))\,du \\
&=\sum_{n=0}^\infty\int_0^{A^q}u^{n/q}(1+\log(u))\,du \\
&=\sum_{n=0}^\infty\br{\int_0^{A^q}u^{n/q}\,du+\int_0^{A^q}u^{n/q}\log(u)\,du} \\
&=\sum_{n=0}^\infty\br{\frac{qA^{n+q}}{n+q}+\frac{qA^{n+q}\log(A^q)}{n+q}-\frac{q^2A^{n+q}}{(n+q)^2}} \\
&=A^q\br{(1+\log(A^q))\sum_{n=0}^\infty\frac{qA^n}{n+q}-\sum_{n=0}^\infty\frac{q^2A^n}{(n+q)^2}} \\
&=A^q\br{(1+\log(A^q))\br{1+\frac{qA}{q+1}+qA^2\sum_{n=2}^\infty\frac{A^{n-2}}{n+q}}-1-\frac{q^2A}{(q+1)^2}-\sum_{n=2}^\infty\frac{q^2A^n}{(n+q)^2}} \\
&\leq A^q\br{(1+\log(A^q))\br{1+\frac{qA}{q+1}+2qA^2(\log(4)-1)}-1-\frac{q^2A}{(q+1)^2}} \\
&\leq A^q\br{(1+q\log(A))2A(1+qA(\log(4)-1))-1-\frac{q^2}{(q+2)^2}} \\
&\leq A^q\br{\br{1-\frac{2q\log^2(2)}{q+\log(4)}}(1+q\log(2))-1-\frac{q^2}{(q+2)^2}}.
\end{align*}
In the first inequality we removed non-positive terms and applied Lemma~\ref{lemma:bound_polygamma}. In the second we applied the facts that $1+\frac{qA}{q+1}=2A$ and $\frac{q^2}{(q+2)^2}\leq\frac{q^2A}{(q+1)^2}$. The last inequality follows from Lemma~\ref{lemma:bound_log_A} and Lemma~\ref{lemma:bound_1_qAB}.

This quantity is now non-positive if
\begin{align*}
\br{1-\frac{2q\log^2(2)}{q+\log(4)}}(1+q\log(2))-1-\frac{q^2}{(q+2)^2}&\leq 0.
\end{align*}
This holds at equality for $q=0$, so it suffices to show that the derivative is non-positive. Writing $x=\log(2)$, abbreviating by $\cdots$ the left-hand-side of the previous inequality, we have
\begin{align*}
\frac{\partial}{\partial q}\br{\cdots}&=\frac{\partial}{\partial q}\br{\br{1-\frac{2q\log^2(2)}{q+\log(4)}}(1+q\log(2))-1-\frac{q^2}{(q+2)^2}} \\
&=\frac{\partial}{\partial q}\br{q\log(2)-\frac{2q\log^2(2)}{q+\log(4)}(1+q\log(2))-\frac{q^2}{(q+2)^2}} \\
&=-\frac{\partial}{\partial q}\br{2x^2\frac{q}{q+2x}+2x^3\frac{q^2}{q+2x}}+x-\frac{2q(q+2)^2-2q^2(q+2)}{(q+2)^4} \\
&=-2x^2\br{\frac{q+2x-q}{(q+2x)^2}}-2x^3\br{\frac{2q(q+2x)-q^2}{(q+2x)^2}}+x-\frac{4q}{(q+2)^3} \\
&=\frac{(q+2x)^2x-4x^3-2x^3q(q+4x)}{(q+2x)^2}-\frac{4q}{(q+2)^3} \\
&=\frac{xq^2+4qx^2+4x^3-4x^3-2x^3q^2-8x^4q}{(q+2x)^2}-\frac{4q}{(q+2)^3} \\
&=\frac{xq(q+4x-2x^2(q+4x))}{(q+2x)^2}-\frac{4q}{(q+2)^3} \\
&=\frac{qx(1-2x^2)(q+4x)}{(q+2x)^2}-\frac{4q}{(q+2)^3} \\
&=q\br{\frac{\log(2)\br{1-2\log^2(2)}(q+\log(16))}{(q+\log(4))^2}-\frac{4}{(q+2)^3}} \\
&\leq q\br{\frac{\log(2)\br{1-2\log^2(2)}(1+\log(16))}{(0+\log(4))^2}-\frac{4}{(1+2)^3}} \\
&\leq -0.08q \\
&\leq 0,
\end{align*}
\end{proof}

\begin{lemma}\label{lemma:bound_log_A}
For $q\in[0,1]$ and $A=\frac{q+1}{q+2}$,
\begin{align*}
-\frac{\log(2)}{1+q\log(2)}\leq\log(A)\leq -\frac{2\log^2(2)}{q+2\log(2)}.
\end{align*}
\end{lemma}
\begin{proof}
For the first inequality from the thesis, we need

\begin{align*}
\log(A)+\frac{\log(2)}{1+q\log(2)}&\geq 0 \\
\iff \br{1+q\log(2)}\log\br{\frac{q+1}{q+2}}+\log(2)&\geq 0.
\end{align*}

One can readily show that the second derivative is negative, so the function is concave and exceeds a line between the endpoints, where the inequality also holds. This means the inequality holds for all $q\in[0,1]$. More in detail, we have
\begin{align*}
\frac{\partial}{\partial q}\sbr{\br{1+q\log(2)}\log\br{\frac{q+1}{q+2}}}&=\log(2)\log\br{\frac{q+1}{q+2}}+\frac{1+q\log(2)}{(q+1)(q+2)}, \\
\frac{\partial^2}{\partial^2 q}\sbr{\br{1+q\log(2)}\log\br{\frac{q+1}{q+2}}}&=\frac{\partial}{\partial q}\br{\log(2)\log\br{\frac{q+1}{q+2}}+\frac{1+q\log(2)}{(q+1)(q+2)}} \\
&=\frac{\log{2}}{(q+1)(q+2)}+\frac{\log(2)(q+1)(q+2)-(1+q\log(2))(2q+3)}{(q+1)^2(q+2)^2} \\
&=\frac{2\log(2)(q+1)(q+2)-(1+q\log(2))(2q+3)}{(q+1)^2(q+2)^2} \\
&=\frac{\log(2)(2q^2+6q+4)-(1+q\log(2))(2q+3)}{(q+1)^2(q+2)^2} \\
&=\frac{\log(2)(3q+4)-(2q+3)}{(q+1)^2(q+2)^2} \\
&=\frac{q(3\log(2)-2)+4\log(2)-3}{(q+1)^2(q+2)^2}.
\end{align*}

The numerator is linear in $q$ and negative for $q=0$ and $q=1$, while the numerator is strictly positive for every $q$. So, the second derivative is negative for all $q$.

The second inequality from the thesis holds at equality for $q=0$. Further, one can show that the derivative of the difference is negative for $q\in[0,1)$:

\begin{align*}
\frac{\partial}{\partial q}\br{\log(A)+\frac{2\log^2(2)}{q+\log(4)}}&=\frac{1}{(q+1)(q+2)}-\frac{2 \log^2(2)}{(q+\log(4))^2}.
\end{align*}

It is not hard to verify that for $q\in[0,1]$

\begin{align*}
(q+1)(q+2)\geq\frac{(q+\log(4))^2}{2\log^2(2)},
\end{align*}

which completes the proof.

\end{proof}

\begin{lemma}\label{lemma:bound_1_qAB}
For $q\in[0,1]$ and $A=\frac{q+1}{q+2}$, we have

\begin{align*}
2A(1+qA(\log(4)-1))\leq 1+q\log(2).
\end{align*}
\end{lemma}
\begin{proof}
\begin{align*}
2A(1+qA(\log(4)-1))&\leq1+q\log(2) \\
\iff \frac{2(1+q)(1+qA(\log(4)-1))}{2+q}-1-q\log(2)&\leq 0 \\
\iff \frac{q^2 (q (\log (8)-2)-3+\log (16))}{(q+2)^2}&\leq 0 \\
\iff q(\log(8)-2)-3+\log(16)&\leq 0.
\end{align*}

This is increasing in $q$ and one can verify this inequality holds for $q=1$.
\end{proof}

\begin{lemma}\label{lemma:bound_polygamma}
For $A=\frac{q+1}{q+2}$, $q\in[0,1]$, we have
\begin{align*}
\sum_{n=2}^\infty\frac{A^{n-2}}{n+q}&\leq2(\log(4)-1).
\end{align*}
\end{lemma}
\begin{proof}
One can verify that the sum is decreasing in $q$, so the left hand side is maximized at $q=0$. Then, using the fact that $\sum_{n=1}^\infty\frac{2^{-n}}{n}=\log(2)$ (this is the polylogarithm of order $1$ evaluated at $1/2$), we have
\begin{align*}
\sum_{n=2}^\infty\frac{A^{n-2}}{n+q}&\leq\sum_{n=2}^\infty\frac{2^{2-n}}{n}=-2+4\sum_{n=1}^\infty\frac{2^{-n}}{n}=2(\log(4)-1).
\end{align*}
\end{proof}

\begin{lemma}\label{lemma:first_ineq}
For $q\in[0,1]$,
\begin{align*}
1+\br{\frac{q}{q+1}}\log\br{q+2}\geq\br{\frac{q+2}{q+1}}^q.
\end{align*}
\end{lemma}
\begin{proof}
We will show this inequality after taking a logarithm on both sides.

We begin with two bounds for the logarithm, for $x\in[0,1]$, we have
\begin{align}
\log(1+x)&\geq\frac{2x}{2+x}, \label{eq:equnoo} \\
\log(2+x)&\geq\log(2)+\frac{x}{2}-\frac{x^2}{4(1+x)}. \label{eq:eqdue}
\end{align}
Both are easy to verify as they hold for $x=0$ and the difference of the left and the right side has positive derivative on $x\in[0,1]$.
In fact, the first holds for all $x\geq 0$ (note $2x/(2+x)=x-x^2/(2+x)$):

\begin{align*}
\frac{\partial}{\partial x}\br{\log(1+x)-x+\frac{x^2}{2+x}}&=\frac{1}{1+x}-1+\frac{2x(2+x)-x^2}{(2+x)^2} \\
&=\frac{1}{1+x}-1+\frac{(4+x)x}{(2+x)^2} \\
&=\frac{-x(2+x)^2+(4+x)x(1+x)}{(1+x)(2+x)^2} \\
&=\frac{x^2}{(1+x)(2+x)^2}.
\end{align*}

The second for $x\in[0,\sqrt{2}]$:

\begin{align*}
\frac{\partial}{\partial x}\br{\log(2+x)-\log(2)-\frac{x}{2}+\frac{x^2}{4(1+x)}}&=\frac{1}{2+x}-\frac{1}{2}+\frac{8x(1+x)-4x^2}{16(1+x)^2} \\
&=\frac{x(2-x^2)}{4(1+x)^2(2+x)}.
\end{align*}

Applying~\eqref{eq:equnoo} and~\eqref{eq:eqdue} to the logarithm of the left hand side, we have

\begin{align*}
\log\br{1+\br{\frac{q}{q+1}}\log\br{q+2}}&\geq\log\br{1+\br{\frac{q}{q+1}}\br{\log(2)+\frac{q}{2}-\frac{q^2}{4(q+1)}}} \\
&=\log\br{1+\frac{4q(q+1)\log(2)+q^2(q+2)}{4(q+1)^2}} \\
&\geq2\times\frac{4q(q+1)\log(2)+q^2(q+2)}{4(q+1)^2}\times\frac{4(q+1)^2}{4q(q+1)\log(2)+q^2(q+2)+8(q+1)^2} \\
&=\frac{8q(q+1)\log(2)+2q^2(q+2)}{4q(q+1)\log(2)+q^2(q+2)+8(q+1)^2}.
\end{align*}

Next apply Lemma~\ref{lemma:bound_log_A}, to get

\begin{align*}
\frac{q\log(2)}{1+q\log(2)}&\geq q\log\br{\frac{q+2}{q+1}}.
\end{align*}

The proof now hinges on showing that

\begin{align*}
\frac{8q(q+1)\log(2)+2q^2(q+2)}{4q(q+1)\log(2)+q^2(q+2)+8(q+1)^2}&\geq\frac{q\log(2)}{1+q\log(2)}.
\end{align*}
Via simple algebraic manipulation, one can show that this holds if and only if
\begin{align*}
q^2\log(2)-2(1-\log(2))(2\log(2)-1)q+4(\log(2)-1)^2&\geq 0.
\end{align*}
This is now a quadratic in $q$, and in fact it has negative discriminant, so it has no real roots. Further, it is positive at $q=0$, so the quadratic is positive for all $q\in[0,1]$.
\end{proof}

\section{Plots for Numerical Experiments}\label{appx:numerical-experiments}

\begin{figure}
    \centering
    \includegraphics[width=1\linewidth]{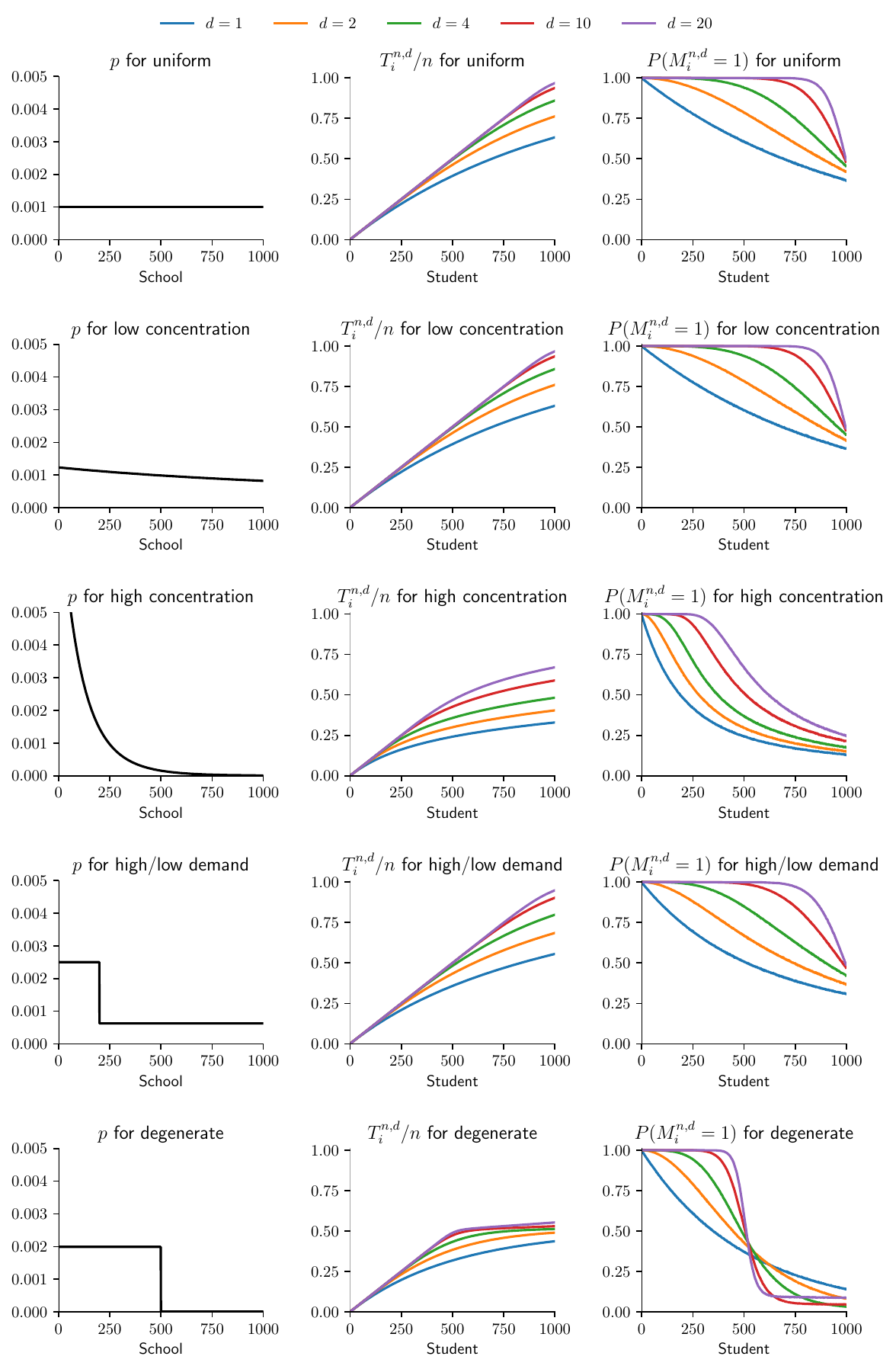}
    \caption{}\label{fig:nonuniform-numerics}
\end{figure}

\begin{figure}
    \centering
    \includegraphics[width=1\linewidth]{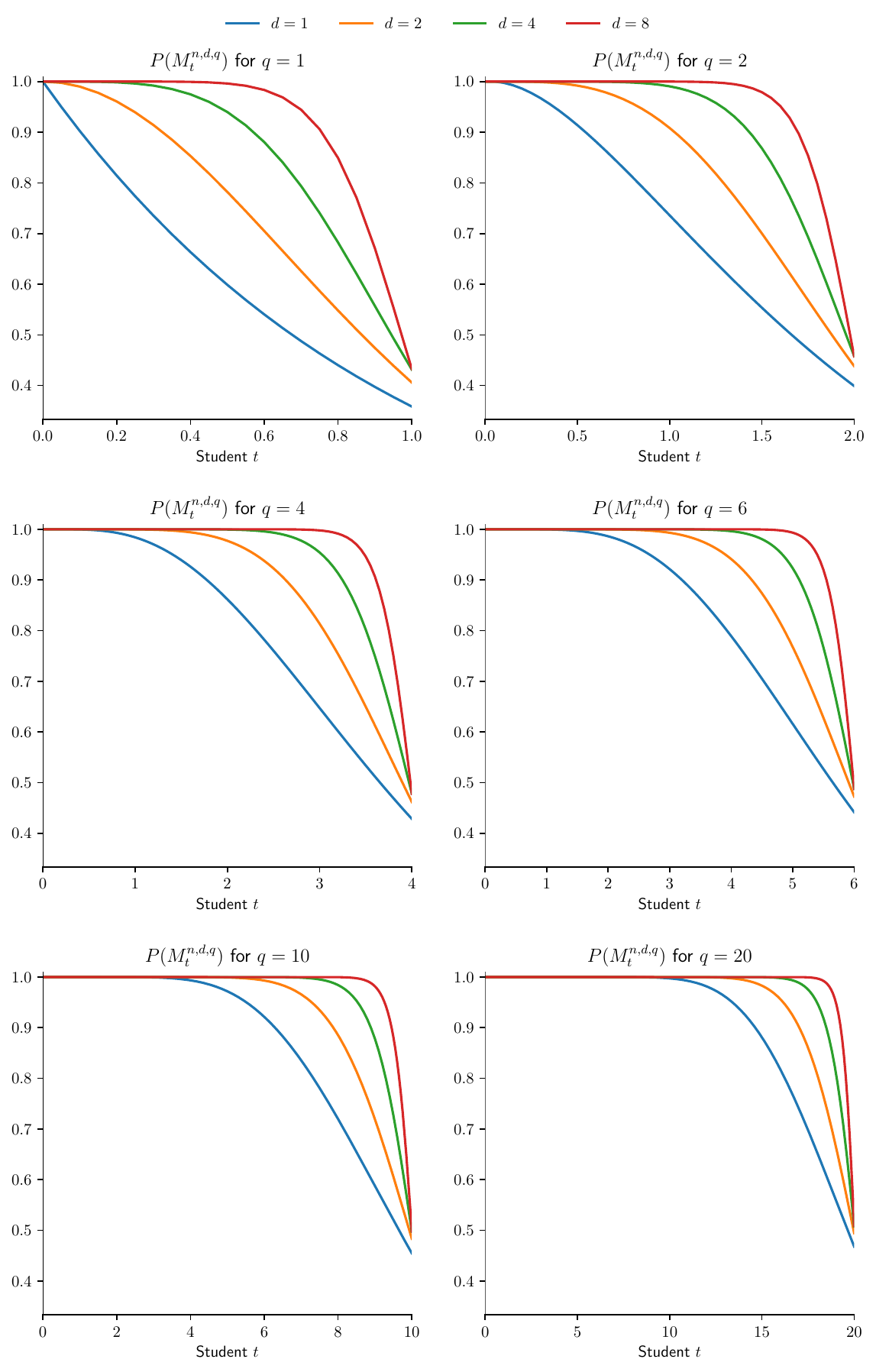}
    \label{fig:enter-label}
\end{figure}

\newpage\end{document}